\documentclass[twoside,11pt]{article}

\usepackage{melba}

\usepackage{bm}
\usepackage{booktabs}
\usepackage{graphicx}
\usepackage{multirow}
\usepackage{float}
\usepackage{amsmath,amsfonts}

\usepackage{xcolor}
\usepackage[normalem]{ulem}

\newtheorem{prop}{Proposition}

\DeclareMathOperator*{\argmax}{arg\,max}
\DeclareMathOperator*{\argmin}{arg\,min}
\DeclareMathOperator{\Forall}{\forall}

\melbaheading{008}{https://www.melba-journal.org/article/22528}{2021}{1--26}{08/2020}{04/2021}{Laves, Ihler, Fast, Kahrs, Ortmaier}{Medical Imaging with Deep Learning (MIDL) 2020}{Marleen de Bruijne, Tal Arbel, Ismail Ben Ayed, Hervé Lombaert}

\ShortHeadings{Recalibration of Aleatoric and Epistemic Regression Uncertainty}{Laves, Ihler, Fast, Kahrs, Ortmaier}
\firstpageno{1}

\title{Recalibration of Aleatoric and Epistemic\\ Regression Uncertainty in Medical Imaging}

\author{\name Max-Heinrich Laves \email max-heinrich.laves@tuhh.de \\
    \addr Institute of Medical Technology and Intelligent Systems, Hamburg University of Technology, Hamburg, Germany\\
    Institute of Mechatronic Systems, Leibniz Universität Hannover, Hannover, Germany
	\AND
	\name Sontje Ihler \email ihler@imes.uni-hannover.de \\
	\addr Institute of Mechatronic Systems, Leibniz Universität Hannover, Hannover, Germany
	\AND
	\name Jacob F.\ Fast \email fast@imes.uni-hannover.de \\
	\addr Institute of Mechatronic Systems, Leibniz Universität Hannover, Hannover, Germany\\
	Department of Phoniatrics and Pediatric Audiology, Hannover Medical School, Hannover, Germany
	\AND
	\name Lüder A. Kahrs \email lakahrs@cs.toronto.edu \\
	\addr Centre for Image Guided Innovation and Therapeutic Intervention, The Hospital for Sick Children, Toronto, Canada \\
	\addr Department of Mathematical and Computational Sciences, University of Toronto Mississauga, Mississauga, Canada
	\AND
	\name Tobias Ortmaier \email ortmaier@imes.uni-hannover.de \\
	\addr Institute of Mechatronic Systems, Leibniz Universität Hannover, Hannover, Germany\\
}

\begin{document}

\maketitle

\begin{abstract}
The consideration of predictive uncertainty in medical imaging with deep learning is of utmost importance.
We apply estimation of both aleatoric and epistemic uncertainty by variational Bayesian inference with Monte Carlo dropout to regression tasks and show that predictive uncertainty is systematically underestimated.
We apply $ \sigma $ scaling with a single scalar value; a simple, yet effective calibration method for both types of uncertainty.
The performance of our approach is evaluated on a variety of common medical regression data sets using different state-of-the-art convolutional network architectures.
In our experiments, $ \sigma $ scaling is able to reliably recalibrate predictive uncertainty.
It is easy to implement and maintains the accuracy.
Well-calibrated uncertainty in regression allows robust rejection of unreliable predictions or detection of out-of-distribution samples.
Our source code is available at:
\href{https://github.com/mlaves/well-calibrated-regression-uncertainty}{github.com/mlaves/well-calibrated-regression-uncertainty}
\end{abstract}

\begin{keywords}
Bayesian approximation, variational inference
\end{keywords}

\section{Introduction}
\label{sec:intro}

Predictive uncertainty should be considered in any medical imaging task that is approached with deep learning.
Well-calibrated uncertainty is of great importance for decision-making and is anticipated to increase patient safety.
It allows to robustly reject unreliable predictions or out-of-distribution samples.
In this paper, we address the problem of miscalibration of regression uncertainty with application to medical image analysis.

For the task of regression, we aim to estimate a continuous target value $ \bm{y} \in \mathbb{R}^{d} $ given an input image $ \bm{x} $.
Regression in medical imaging with deep learning has been applied to forensic age estimation from hand CT/MRI \citep{Halabi2018,Stern2016}, natural landmark localization \citep{Payer2019}, cell detection in histology \citep{Xie2018}, or instrument pose estimation \citep{Gessert2018}.
By predicting the coordinates of object boundaries, segmentation can also be performed as a regression task.
This has been done for segmentation of pulmonary nodules in CT \citep{Messay2015}, kidneys in ultrasound \citep{Yin2019}, or left ventricles in MRI \citep{Tan2017}.
In registration of medical images, a continuous displacement field is predicted for each coordinate of $ \bm{x} $, which has also recently been addressed by CNNs for regression \citep{Dalca2019}.

In medical imaging, it is crucial to consider the predictive uncertainty of deep learning models.
Bayesian neural networks (BNN) and their approximation provide mathematical tools for  reasoning the uncertainty \citep{Bishop2006,Kingma2013}.
In general, predictive uncertainty can be split into two types: aleatoric and epistemic uncertainty \citep{Tanno2017,Kendall2017}.
This distinction was first made in the context of risk management \citep{Hora1996}.
Aleatoric uncertainty arises from the data directly; e.\,g.\ sensor noise or motion artifacts.
Epistemic uncertainty is caused by uncertainty in the model parameters due to a limited amount of training data \citep{Bishop2006}.
A well-accepted approach to quantify epistemic uncertainty is variational inference with Monte Carlo (MC) dropout, where dropout is used at test time to sample from the approximate posterior \citep{Gal2016}.

\begin{figure}
    \centering
    \includegraphics[scale=0.9]{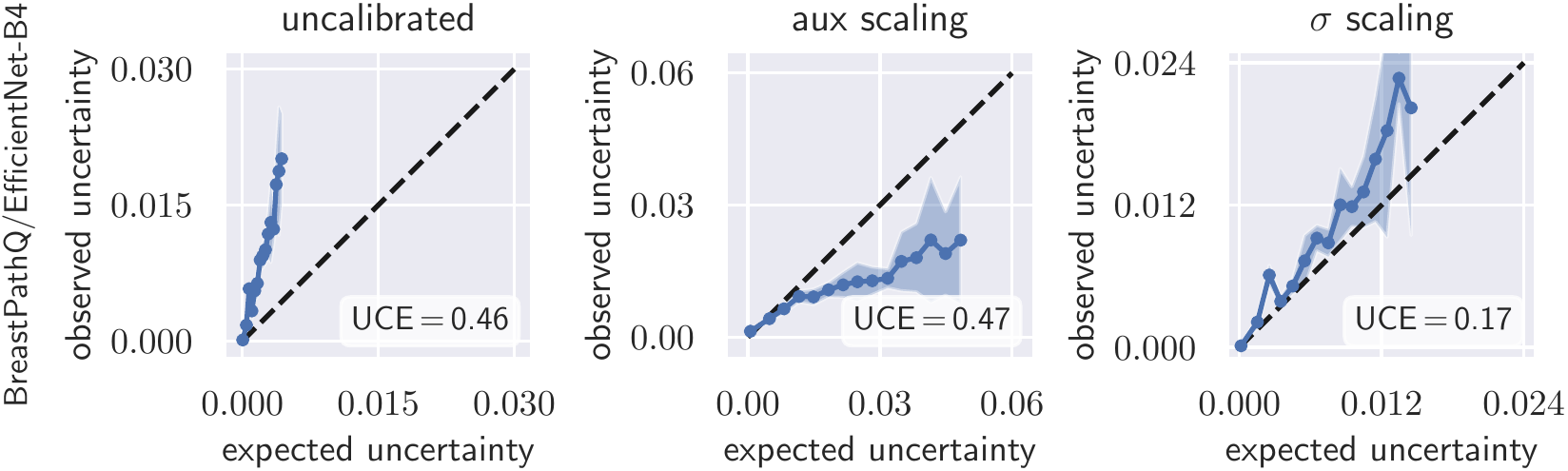}
    \caption{Calibration plots and uncertainty calibration error (UCE) for EfficientNet-B4 on BreastPathQ test set. Uncalibrated uncertainty is underestimated and does not correspond well with the model error (left). Uncertainty can be calibrated most effectively with $ \sigma $ scaling (right). Solid lines show the mean and shaded areas show standard deviation from 5 repeated runs. Dashed lines denote perfect calibration.}
    \label{fig:opener}
\end{figure}

Uncertainty quantification in regression problems in medical imaging has been addressed by prior work.
Medical image enhancement with image quality transfer (IQT) has been extended to a Bayesian approach to obtain pixel-wise uncertainty \citep{Tanno2016}.
Additionally, CNN-based IQT was used to estimate both aleatoric and epistemic uncertainty in MRI super-resolution \citep{Tanno2017}.
\citet{Dalca2019} estimated uncertainty for a deformation field in medical image registration using a probabilistic CNN.
Registration uncertainty has also been addressed outside the deep learning community \citep{Luo2019}.
\citet{Schlemper2018} used sub-network ensembles to obtain uncertainty estimates in cardiac MRI reconstruction.
Aleatoric and epistemic uncertainty was also used in multitask learning for MRI-based radiotherapy planning \citep{Bragman2018}.


Uncertainty obtained by deep BNNs tends to be miscalibrated, i.\,e.\ it does not correlate well with the model error\,\citep{Laves2019NIPS}.
Fig.\,\ref{fig:opener} shows calibration plots (observed uncertainty vs.\ expected uncertainty) for uncalibrated and calibrated uncertainty estimates.
The predicted uncertainty (taking into account both epistemic and aleatoric uncertainty) is underestimated and does not allow robust detection of uncertain predictions at test time.

Calibration of uncertainty in regression has been addressed in prior work outside medical imaging.
In \citep{Kuleshov2018}, inaccurate uncertainties from Bayesian models for regression are recalibrated using a technique inspired by Platt scaling.
Given a pre-trained, miscalibrated model $ \bm{H} $, an auxiliary model $ \bm{R} : [0,1]^{d} \rightarrow [0,1]^{d} $ is trained, that yields a calibrated regressor $ \bm{R} \circ \bm{H} $.
In \citep{Phan2018}, this method was applied to bounding box regression.
However, an auxiliary model with enough capacity will always be able to recalibrate, even if the predicted uncertainty is completely uncorrelated with the real uncertainty.
Furthermore, \citet{Kuleshov2018} state that calibration via $ \bm{R} $ is possible if enough independent and identically distributed (i.i.d.) data is available.
In medical imaging, large data sets are usually hard to obtain, which can cause $ \bm{R} $ to overfit the calibration set.
This downside was addressed in \citep{Levi2019}, which is most related to our work.
They proposed to scale the standard deviation of a Gaussian model to recalibrate aleatoric uncertainty.
In contrast to our work, they do not take into account epistemic uncertainty, which is an important source of uncertainty, especially when dealing with small data sets in medical imaging.

This paper extends a preliminary version of this work presented at the Medical Imaging with Deep Learning (MIDL) 2020 conference \citep{Laves2020}.
We continue this work by providing a new derivation of our definition of perfect calibrtaion, new experimental results, analysis and discussion.
Additionally, prediction intervals are computed to further assess the quality of the estimated uncertainty.
We find that prediction intervals are estimated too narrow and that recalibration can mitigate this problem.

To the best of our knowledge, calibration of predictive uncertainty for regression tasks in medical imaging has not been addressed.
Our main contributions are:
(1) We suggest to use $ \sigma $ scaling in a separate calibration phase to tackle underestimation of aleatoric and epistemic uncertainty,
(2) we propose to use the uncertainty calibration error and prediction intervals to assess the quality of the estimated uncertainty, and
(3) we perform extensive experiments on four different data sets to show the effectiveness of the proposed method.

\section{Methods}

In this section, we discuss estimation of aleatoric and epistemic uncertainty for regression and show why uncertainty is systematically miscalibrated.
We propose to use $ \sigma $ scaling to jointly calibrate aleatoric and epistemic uncertainty.

\subsection{Conditional Log-Likelihood for Regression}
\label{sec:cond-log-likelihood}

We revisit regression under the maximum posterior (MAP) framework to derive direct estimation of heteroscedastic aleatoric uncertainty.
That is, the aleatoric uncertainty varies with the input and is not assumed to be constant.
The goal of our regression model is to predict a target value $ \bm{y} $ given some new input $ \bm{x} $ and a training set $ \mathcal{D} $ of $ m $ inputs $ \bm{X} = \{ \bm{x}_1, \ldots, \bm{x}_m \} $ and their corresponding (observed) target values $ \bm{Y} = \{ \bm{y}_1, \ldots , \bm{y}_m \} $.
We assume that $ \bm{y} $ has a Gaussian distribution $ \mathcal{N} \left( \bm{y} ; \hat{\bm{y}}(\bm{x}), \hat{\sigma}^{2}(\bm{x}) \right) $ with mean equal to $ \hat{\bm{y}}(\bm{x}) $ and variance $ \hat{\sigma}^{2}(\bm{x}) $.
A neural network with parameters $ \bm{\theta} $
\begin{equation}
    \bm{f}_{\bm{\theta}} \left( \bm{x} \right) = \left[ \hat{\bm{y}}(\bm{x}), \hat{\sigma}^{2} (\bm{x}) \right] , ~ \hat{\bm{y}} \in \mathbb{R}^{d}, ~ \hat{\sigma}^{2} \in \mathbb{R}, \hat{\sigma}^{2} \geq 0
    \label{eq:neural_net}
\end{equation}
outputs these values for a given input \citep{Nix1994}.
By assuming a Gaussian prior over the parameters $ \bm{\theta} \sim \mathcal{N}(\bm{\theta} ; \bm{0}, \lambda^{-1} \bm{I}) $, MAP estimation becomes maximum-likelihood estimation with added weight decay \citep{Bishop2006}.
With $ m $ i.i.d.\ random samples, the conditional log-likelihood $ \log p(\bm{Y} \,\vert\, \bm{X}, \bm{\theta}) $ is given by
\begin{align} 
    \log p(\bm{Y} \,\vert\, \bm{X}, \bm{\theta}) =& \sum_{i=1}^{m} \log \left( \frac{1}{\sqrt{2\pi} \hat{\sigma}^{(i)}_{\bm{\theta}}} \exp \left\{ - \frac{\big\Vert \bm{y}^{(i)} - \hat{\bm{y}}_{\bm{\theta}}^{(i)} \big\Vert^{2}}{2 \big( \hat{\sigma}^{(i)}_{\bm{\theta}} \big)^{2} } \right\} \right) \\
    = & - \dfrac{m}{2} \log \left( 2\pi \right) - \sum_{i=1}^{m} \log \big( \hat{\sigma}^{(i)}_{\bm{\theta}} \big) + \frac{1}{2 \big( \hat{\sigma}_{\bm{\theta}}^{(i)} \big)^{2} } \big\Vert \bm{y}^{(i)} - \hat{\bm{y}}_{\bm{\theta}}^{(i)} \big\Vert^{2} ~ .
        \label{eq:gaussian_derive}
\end{align}
The dependence on $ \bm{x} $ has been omitted to simplify the notation. 
Maximizing the log-likelihood in Eq.\,(\ref{eq:gaussian_derive}) w.r.t.\ $ \bm{\theta} $ is equivalent to minimizing the negative log-likelihood (NLL), which leads to the following optimization criterion (with weight decay)
\begin{equation}
    \mathcal{L}_{\mathrm{G}}(\bm{\theta}) = \sum_{i=1}^{m} \big( \hat{\sigma}^{(i)}_{\bm{\theta}} \big)^{-2} \big\Vert \bm{y}^{(i)} - \hat{\bm{y}}_{\bm{\theta}}^{(i)} \big\Vert^{2} + \log \big( ( \hat{\sigma}_{\bm{\theta}}^{(i)} )^{2} \big) ~ .
    \label{eq:loss_gaussian}
\end{equation}
Here, $ \hat{\bm{y}}_{\bm{\theta}} $ and $ \hat{\sigma}_{\bm{\theta}} $ are estimated jointly by finding $ \bm{\theta} $ that minimizes Eq.\,(\ref{eq:loss_gaussian}).
This can be achieved using gradient descent in a standard training procedure.
In this case, $ \hat{\sigma}_{\bm{\theta}} $ captures the uncertainty that is inherent in the data (aleatoric uncertainty).
To avoid numerical instability due to potential division by zero, we directly estimate $ \log \hat{\sigma}^{2} (\bm{x}) $ and implement Eq.\,(\ref{eq:loss_gaussian}) in similar practice to \citet{Kendall2017}.

\subsection{Biased estimation of \texorpdfstring{$ \sigma $}{Sigma}}

\begin{figure}
    \centering
    \includegraphics[scale=0.9]{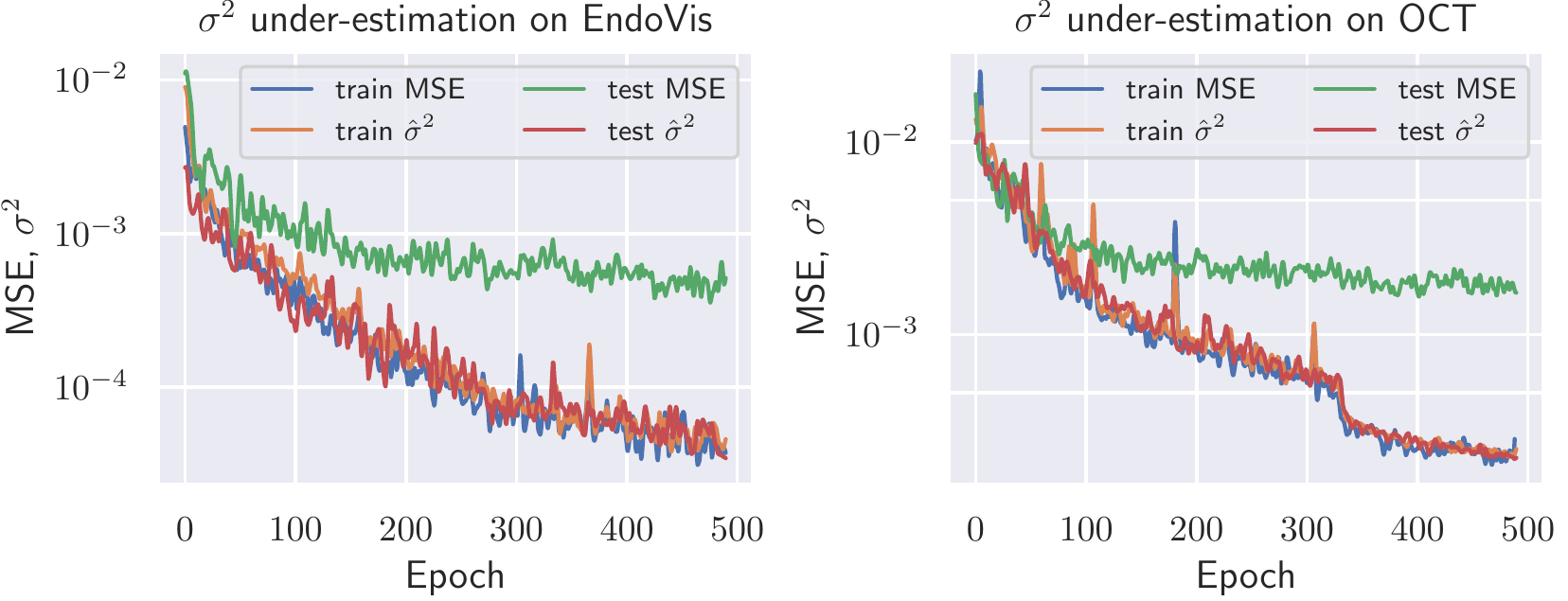}
    \caption{Biased estimation of aleatoric uncertainty $ \sigma^{2} $. The deep model overfits estimation of $ \bm{y} $ on the training set. On unseen test data, the MSE of predictive mean is higher and $ \sigma^{2} $ is underestimated. Early stopping (e.g.\ at epoch 50) would result in an unbiased estimator, but this would not be optimal in terms of test MSE.}
    \label{fig:biased_estimation}
\end{figure}
Ignoring their dependence through $ \bm{\theta} $, the solution to Eq.\,(\ref{eq:loss_gaussian}) decouples estimation of $ \hat{\bm{y}} $ and $ \hat{\sigma} $.
In case of a Gaussian likelihood, minimizing Eq.\,(\ref{eq:loss_gaussian}) w.r.t.\ $ \hat{\bm{y}}^{(i)} $ yields
\begin{equation}
    \hat{\bm{y}}^{(i)} = \argmin_{\hat{\bm{y}}^{(i)}} \mathcal{L}_{\mathrm{G}} = \bm{y}^{(i)} ~ \Forall i ~ .
\end{equation}
Minimizing (\ref{eq:loss_gaussian}) w.r.t. $ (\hat{\sigma}^{(i)})^{2} $ yields
\begin{equation}
    \big( \hat{\sigma}^{(i)} \big)^{2} = \argmin_{ ( \hat{\sigma}^{(i)} )^{2}} \mathcal{L}_{\mathrm{G}} = \Vert \bm{y}^{(i)} - \hat{\bm{y}}^{(i)} \Vert^{2} ~ \Forall i ~ .
    \label{eq:perfect_argmin}
\end{equation}
That is, estimation of $ \sigma^{2} $ should perfectly reflect the squared error.
However, in Eq.\,(\ref{eq:perfect_argmin}) $ \sigma^{2} $ is estimated relative to the estimated mean $ \hat{\bm{y}} $ and therefore biased.
In fact, the maximum likelihood solution systematically underestimates $ \sigma^{2} $, which is a phenomenon of overfitting the training set \citep{Bishop2006}.
The squared error $ \Vert \bm{y} - \hat{\bm{y}} \Vert^{2} $ will be lower on the training set and $ \hat{\sigma}^{2} $ on new samples will be systematically too low (see Fig.\,\ref{fig:biased_estimation}).
This is a problem especially in deep learning, where large models have millions of parameters and tend to overfit.
To solve this issue, we introduce a simple learnable scalar parameter $ s $ to rescale the biased estimation of $ \sigma^{2} $.

\subsection{\texorpdfstring{$ \sigma $}{Sigma} Scaling for Aleatoric Uncertainty}
\label{sec:scaling}

We first derive $ \sigma $ scaling for aleatoric uncertainty.
Using a Gaussian model, we scale the standard deviation $ \sigma $ with a scalar value $ s $ to recalibrate the probability density function
\begin{equation}
    p \left( \bm{y} \vert \bm{x} ; \hat{\bm{y}} (\bm{x}), \hat{\sigma}^{2}(\bm{x})  \right) = \mathcal{N} \left( \bm{y} ; \hat{\bm{y}} (\bm{x}), (s \cdot \hat{\sigma}(\bm{x}))^{2} \right) ~ .
\end{equation}
This results in the following minimization objective:
\begin{equation}
    \mathcal{L}_{\mathrm{G}}(s) = m \log(s) + \tfrac{1}{2} s^{-2} \sum_{i=1}^{m} \big( \hat{\sigma}_{\bm{\theta}}^{(i)} \big)^{-2} \big\Vert \bm{y}^{(i)} - \hat{\bm{y}}_{\bm{\theta}}^{(i)} \big\Vert^{2} ~ .
    \label{eq:loss_scaling_gaussian}
\end{equation}
Eq.\,(\ref{eq:loss_scaling_gaussian}) is optimized w.r.t.\ $ s $ with fixed $ \bm{\theta} $ using gradient descent in a separate calibration phase after training to calibrate aleatoric uncertainty measured by $ \hat{\sigma}_{\bm{\theta}}^{2} $.
In case of a single scalar, the solution to Eq.\,(\ref{eq:loss_scaling_gaussian}) can also be written in closed form as
\begin{equation}
    s = \pm \sqrt{\frac{1}{m} \sum_{i=1}^{m} \big( \hat{\sigma}_{\bm{\theta}}^{(i)} \big)^{-2} \big\Vert \bm{y}^{(i)} - \hat{\bm{y}}_{\bm{\theta}}^{(i)} \big\Vert^{2}} ~ .
\end{equation}
We apply $ \sigma $ scaling to jointly calibrate aleatoric and epistemic uncertainty in the next section.

\subsection{Well-Calibrated Estimation of Predictive Uncertainty}
\label{sec:well}

So far we have assumed a MAP point estimate for $ \bm{\theta} $ which does not consider uncertainty in the parameters.
To quantify both aleatoric and epistemic uncertainty, we extend $ \bm{f}_{\bm{\theta}} $ into a fully Bayesian model under the variational inference framework with Monte Carlo dropout \citep{Gal2016}.
In MC dropout, the model $ \bm{f}_{\tilde{\bm{\theta}}} $ is trained with dropout \citep{Srivastava2014} and dropout is applied at test time by performing $ N $ stochastic forward passes to sample from the approximate Bayesian posterior $ \tilde{\bm{\theta}} \sim q (\bm{\theta}) $.
Following \citep{Kendall2017}, we use MC integration to approximate the predictive variance
\begin{equation}
    \hat{\Sigma}^{2} = \underbrace{ \frac{1}{N} \sum_{n=1}^{N} \left( \hat{\bm{y}}_{n} - \frac{1}{N} \sum_{n=1}^{N} \hat{\bm{y}}_{n} \right)^{2}}_{\mathrm{epistemic}} + \underbrace{ \frac{1}{N} \sum_{n=1}^{N} \hat{\sigma}^{2}_{n} }_{\mathrm{aleatoric}}
    \label{eq:pred_variance}
\end{equation}
and use $ \hat{\Sigma}^{2} $ as a measure of predictive uncertainty.
If the neural network has multiple outputs ($ d > 1 $), the predictive variance is calculated per output and the mean across $ d $ forms the final uncertainty value.
Eq.\,(\ref{eq:pred_variance}) is an unbiased estimator of the approximate predictive variance (see proof in Appendix~\ref{app:proof_var}).
From Eq.~(\ref{eq:proof_bias}) of our proof follows, that $ \hat{\Sigma}^{2} $ is expected to equal the true variance $ \Sigma = \mathbb{E}[(\hat{\bm{y}}-\bm{y})^{2}] $.
Thus, we define perfect calibration of regression uncertainty as
\begin{equation}
    \mathbb{E}_{\bm{x},\bm{y}} \left[ 
    \mathbb{E}[(\hat{\bm{y}}-\bm{y})^{2}] \, \big\vert \, \hat{\Sigma}^{2} = \alpha^{2} \right] = \alpha^{2} \quad \Forall \left\{ \alpha^{2}\in \mathbb{R} \, \vert \, \alpha^{2} \geq 0 \right\} ~ ,
    \label{eq:perfect}
\end{equation}
which extends the definition of \citep{Levi2019} to both aleatoric and epistemic uncertainty.
We expect that additionally accounting for epistemic uncertainty is particularly beneficial for smaller data sets.
However, even in deep learning with Bayesian principles, the approximate posterior predictive distribution can overfit on small data sets.
In practice, this leads to underestimation of the predictive uncertainty.

One could regularize overfitting by early stopping that prevents large differences between training and test loss, which would circumvent underestimation of $ \sigma^{2} $.
However, our experiments show that early stopping is not optimal with regard to accuracy, i.e.\ the squared error of $ \hat{\bm{y}} $ on both training and testing data (see Fig.\,\ref{fig:biased_estimation}).
In contrast, the model with lowest mean error on the validation set underestimates predictive uncertainty considerably.
Therefore, we apply $ \sigma $ scaling to recalibrate the predictive uncertainty $ \hat{\Sigma}^{2} $.
This allows a lower squared error while reducing underestimation of uncertainty as shown experimentally in the following section.

%
%
%
%

\subsection{Expected Uncertainty Calibration Error for Regression}
\label{app:uce}

We extend the definition of miscalibrated uncertainty for classification \citep{Laves2019NIPS} to quantify miscalibration of uncertainty in regression
\begin{equation}
    \mathbb{E}_{\hat{\Sigma}^{2}} \left[ \big\vert \big( \mathbb{E}[(\hat{\bm{y}}-\bm{y})^{2}] \, \big\vert \, \hat{\Sigma}^{2} = \alpha^{2} \big) - \alpha^{2} \big\vert \right] \quad \Forall \left\{ \alpha^{2} \in \mathbb{R} \, \vert \, \alpha^{2} \geq 0 \right\} ~ ,
    \label{eq:uce}
\end{equation}
using the second moment of the error.
On finite data sets, this can be approximated with the expected uncertainty calibration error (UCE) for regression.
Following \citep{Guo2017}, the uncertainty output $ \hat{\Sigma}^{2} $ of a deep model is partitioned into $ K $ bins with equal width.
A weighted average of the difference between the variance and predictive uncertainty is used:
\begin{equation}
    \mathrm{UCE} := {\sum_{k=1}^{K}} \frac{\vert B_{k} \vert}{m} \big\vert {\mathrm{var}}(B_{k}) - \mathrm{uncert}(B_{k}) \big\vert ~ ,
\end{equation}
with number of inputs $ m $ and set of indices $ B_{k} $ of inputs, for which the uncertainty falls into the bin $ k $.
The variance per bin is defined as
\begin{equation}
    \mathrm{var}(B_{k}) := \frac{1}{\vert B_{k} \vert} \sum_{i \in B_{m}} {\frac{1}{N} \sum_{n=1}^{N} \left( \hat{\bm{y}}_{i,n} - \bm{y}_{i} \right)^{2} } ~ ,
\end{equation}
with $ N $ stochastic forward passes, and the uncertainty per bin is defined as
\begin{equation}
    \mathrm{uncert}(B_{k}) := \frac{1}{\vert B_{k} \vert} \sum_{i \in B_{k}} \hat{\Sigma}_{i}^{2} ~ .
\end{equation}
Note that computing the second moment from Eq.\,(\ref{eq:uce}) also incorporates MC samples, which can introduce some bias in the evaluation.
The UCE considers both aleatoric and epistemic uncertainty and is given in \% throughout this work.
Additionally, we plot $ \mathrm{var}(B_{k}) $ vs.\ $ \mathrm{uncert}(B_{k}) $ to create calibration diagrams.

\section{Experiments}
\label{sec:experiments}

We use four data sets and three common deep network architectures to evaluate recalibration with $ \sigma $ scaling.
The data sets were selected to represent various regression tasks in medical imaging with different dimension $ d $ of target value $ \bm{y} \in \mathbb{R}^{d} $:

(1) Estimation of tumor cellularity in histology whole slides of cancerous breast tissue from the BreastPathQ SPIE challenge data set ($ d = 1$) \citep{Martel2019}.
The public data set consists of 2579 images, from which 1379/600/600 are used for training/validation/testing.
The ground truth label is a single scalar $y \in [0, 1] $ denoting the ratio of tumor cells to non-tumor cells.

(2) Hand CT age regression from the RSNA pediatric bone age data set ($ d = 1 $) \citep{Halabi2018}.
The task is to infer a person's age in months from CT scans of the hand.
This data set is the largest used in this paper and has 12,811 images, from which we use 6811/2000/4000 images for training/validation/testing.

(3) Surgical instrument tracking on endoscopic images from the EndoVis endoscopic vision challenge 2015\footnote{\href{https://endovissub-instrument.grand-challenge.org}{endovissub-instrument.grand-challenge.org}} data set ($ d = 2 $).
This data set contains 8,984 video frames from 6 different robot-assisted laparoscopic interventions showing surgical instruments with ground truth pixel coordinates of the instrument's center point $ \bm{y} \in \mathbb{R}^{2} $.
We use 4483/2248/2253 frames for training/validation/testing.
As the public data set is only sparsely annotated, we created our own ground truth labels, which can be found in our code repository.

(4) 6DoF needle pose estimation on optical coherence tomography (OCT) scans from our own data set\footnote{Our OCT pose estimation data set is publicly available at \href{https://github.com/mlaves/3doct-pose-dataset}{github.com/mlaves/3doct-pose-dataset}}.
This data set contains 5,000 3D-OCT scans with the accompanying needle pose $ \bm{y} \in \mathbb{R}^{6} $, from which we use 3300/850/850 for training/validation/testing.
Additional details on this data set can be found in Appendix~\ref{app:dataset}.

All outputs are normalized such that $ \bm{y} \in [0, 1]^{d}$.
The employed network architectures are ResNet-101, DenseNet-201 and EfficientNet-B4 \citep{He2016,Huang2017,Tan2019}.
The last linear layer of all networks is replaced by two linear layers predicting $ \hat{\bm{y}} $ and $ \hat{\sigma}^{2} $ as described in §\,\ref{sec:cond-log-likelihood}.
For MC dropout, we use dropout before the last linear layers.
Dropout is further added after each of the four layers of stacked residual blocks in ResNet.
In DenseNet and EfficientNet, we use the default configuration of dropout during training and testing.
The networks are trained until no further decrease in mean squared error (MSE) on the validation set can be observed.
More details on the training procedure can be found in Appendix \ref{app:training}.

Calibration is performed after training in a separate calibration phase using the validation data set.
We plug the predictive uncertainty $ \hat{\Sigma}^{2} $ into Eq.\,(\ref{eq:loss_scaling_gaussian}) and minimize w.r.t.\ $ s $.
Additionally, we compare $ \sigma $ scaling to a more powerful auxiliary recalibration model $ \bm{R} $ consisting of a two-layer fully-connected network with 16 hidden units and ReLU activations (inspired by \citep{Kuleshov2018}, see §\,\ref{sec:intro}).

\begin{figure}[t]
    \centering
    \includegraphics[scale=0.9]{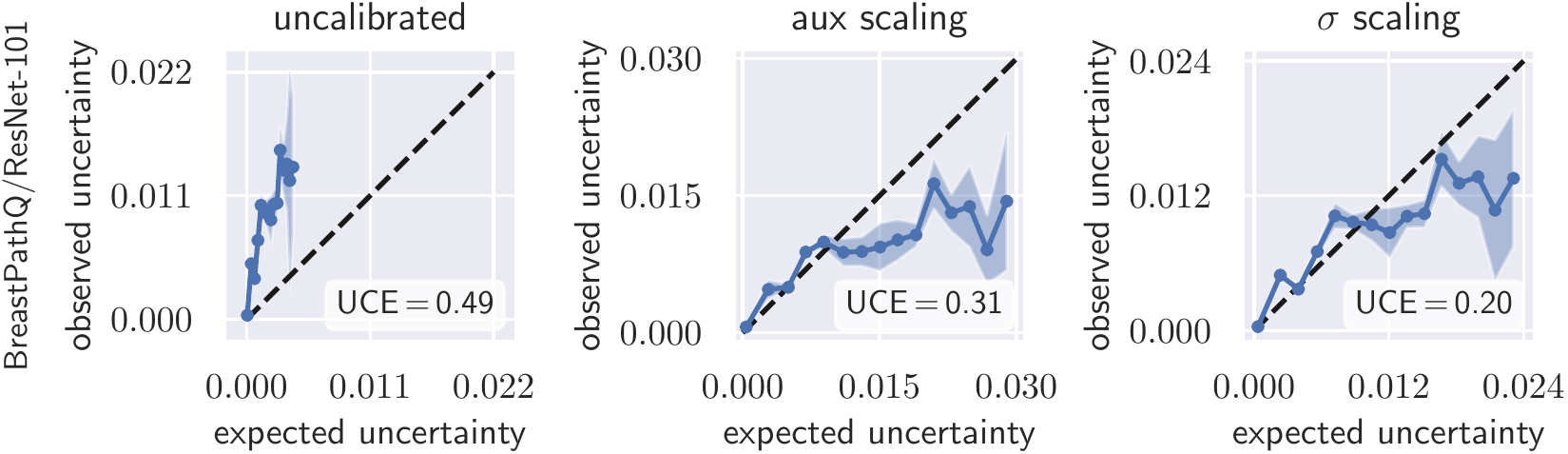} \\
    ~\\
    \includegraphics[scale=0.9]{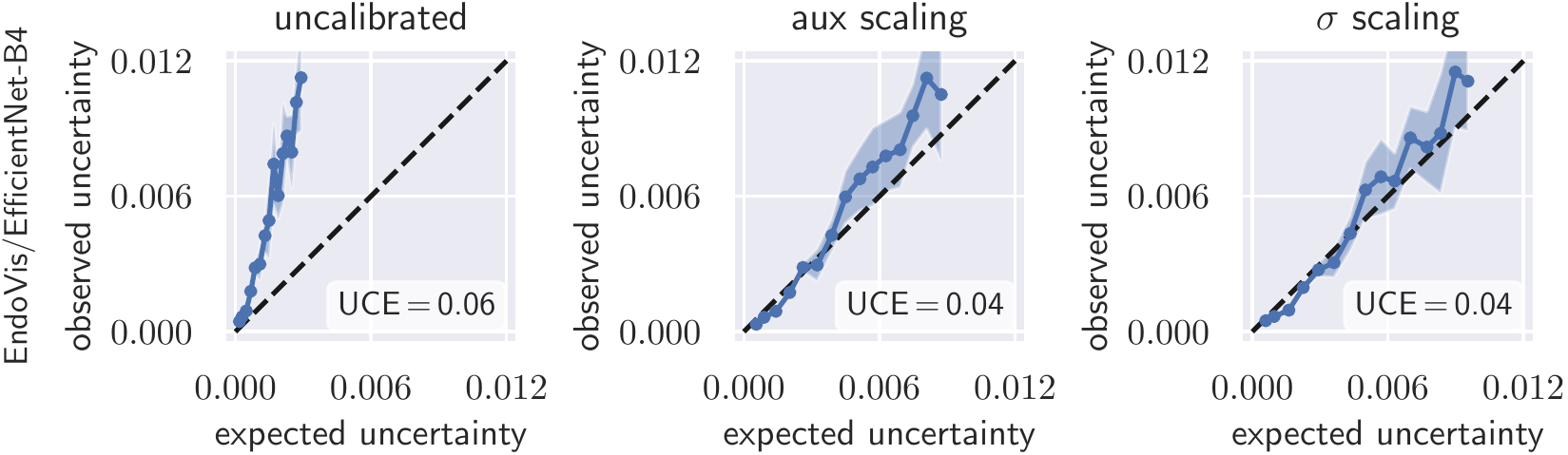} \\
    \caption{Calibration plots for ResNet-101 on BreastPathQ (top row) and EfficientNet-B4 on EndoVis (bottom row). Aux scaling tends to overfit the calibration set, which results in higher UCE compared to simple $ \sigma $ scaling. Dashed lines denote perfect calibration.}
    \label{fig:results}
\end{figure}

\section{Results}
\label{sec:results}

To quantify miscalibration, we use the proposed expected uncertainty calibration error for regression.
We visualize (mis-)calibration in Fig.\,\ref{fig:opener} and Fig.\,\ref{fig:results} using calibration diagrams, which show expected uncertainty vs.\ observed uncertainty.
The discrepancy to the identity function reveals miscalibration.
The calibration diagrams clearly show the underestimation of uncertainty for the uncalibrated models.
After calibration with both aux and $ \sigma $ scaling, the estimated uncertainty better reflects the actual uncertainty.
Figures for all configurations are listed in Appendix \ref{app:figures}.

Tab.\,\ref{tab:results} reports UCE values of all data set/model combinations on the respective test sets.
{The negative log-likelihood also measures miscalibration;
the values on the test set can be found in Tab.\,\ref{tab:results_nll} in the appendix.}
In general, recalibration considerably reduces miscalibration.
On the data sets BoneAge, EndoVis and OCT, both scaling methods perform similarly well.
However, on the BreastPathQ data set, $ \sigma $ scaling clearly outperforms aux scaling in terms of UCE.
BreastPathQ is the smallest data set and thus has the smallest calibration set size.
We hypothesize that the more powerful auxiliary model $ \bm{R} $ overfits the calibration set (see BreastPathQ/DenseNet-201 in Tab.\,\ref{tab:results}), which leads to an increase of UCE on the test set.
An ablation study on BreastPathQ for the auxiliary model can be found in Appendix~\ref{app:aux_ablation}.

We also compare our approach to \citet{Levi2019} in Tab.\,\ref{tab:results}, which only considers aleatoric uncertainty.
The aleatoric uncertainty is well-calibrated if it reflects the bias $ \left( \mathbb{E} \left[\hat{\bm{y}}_{n}\right] - \bm{y} \right)^{2} $, which is given by the squared error between the expectation of the stochastic predictions $ \hat{\bm{y}}_{n} $ and the ground truth. 
Therefore, the UCE for aleatoric-only is computed by $ \mathrm{UCE} = {\sum_{k=1}^{K}} \frac{\vert B_{k} \vert}{m} \big\vert {\mathrm{err}}(B_{k}) - \mathrm{uncert}(B_{k}) \big\vert ~ $, where $ \mathrm{err}(\cdot) $ is the mean squared error and $ \mathrm{uncert}(\cdot) $ is the mean aleatoric uncertainty per bin.
Consideration of epistemic uncertainty is especially beneficial on smaller data sets (BreastPathQ), where our approach outperforms \citet{Levi2019}.
On larger data sets, the benefit diminishes and both approaches are equally calibrated.

Additionally, we report UCE values from a DenseNet ensemble for comparison.
In contrast to what is reported by \citet{Lakshminarayanan2017}, the deep ensemble tends to be calibrated worse.
Only on BoneAge, the ensemble is better calibrated prior to recalibration of the other methods.
After recalibration, both approaches outperform the deep ensemble.

Fig.\,\ref{fig:calibrated_estimation} shows the result of intra-training calibration of aleatoric uncertainty.
It indicates that the gap between training and test loss is successfully closed.

\begin{figure}
    \centering
    \includegraphics[scale=0.9]{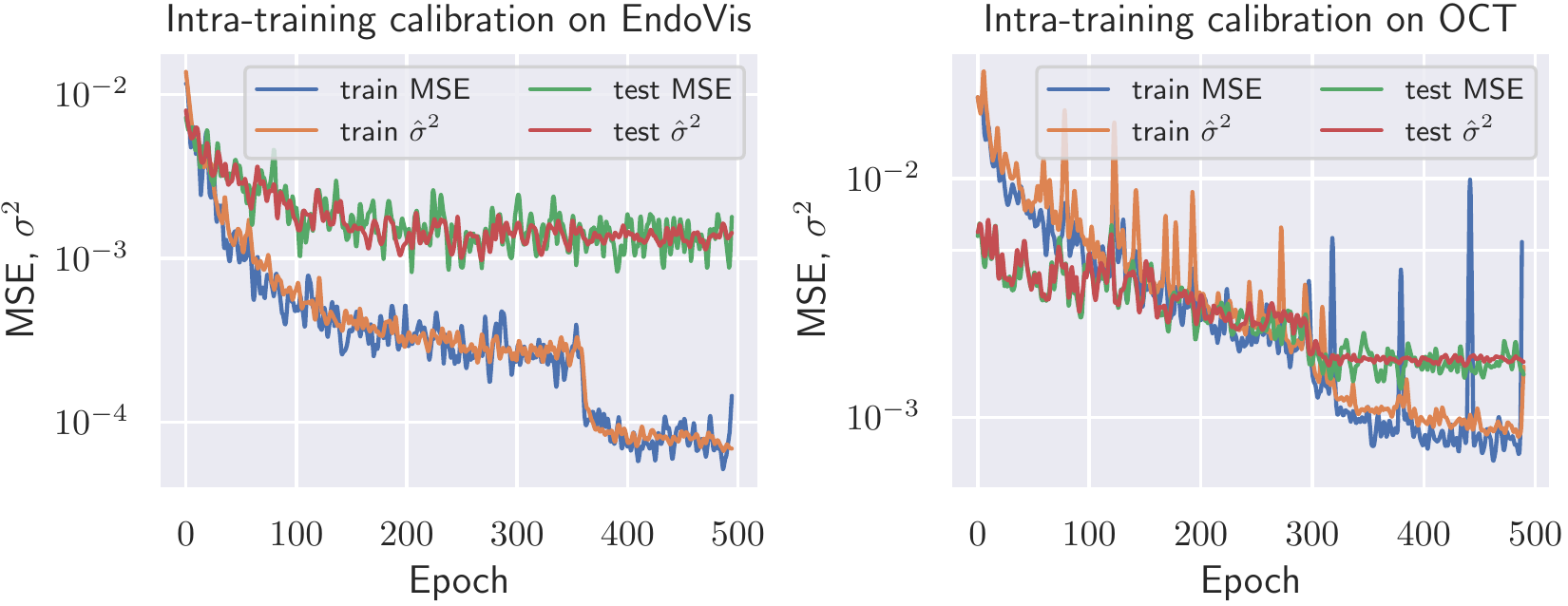}
    \caption{(Left) Intra-training calibration of aleatoric uncertainty with $ \sigma $ scaling. The deep model no longer underestimates $ \hat{\sigma}^{2} $ on unseen test data. (Right) The MSE of predictive mean is higher and $ \sigma^{2} $ is underestimated. Note: Calibration is only applied at test time.}
    \label{fig:calibrated_estimation}
\end{figure}

\begin{figure}
    \centering
    \includegraphics[scale=0.89]{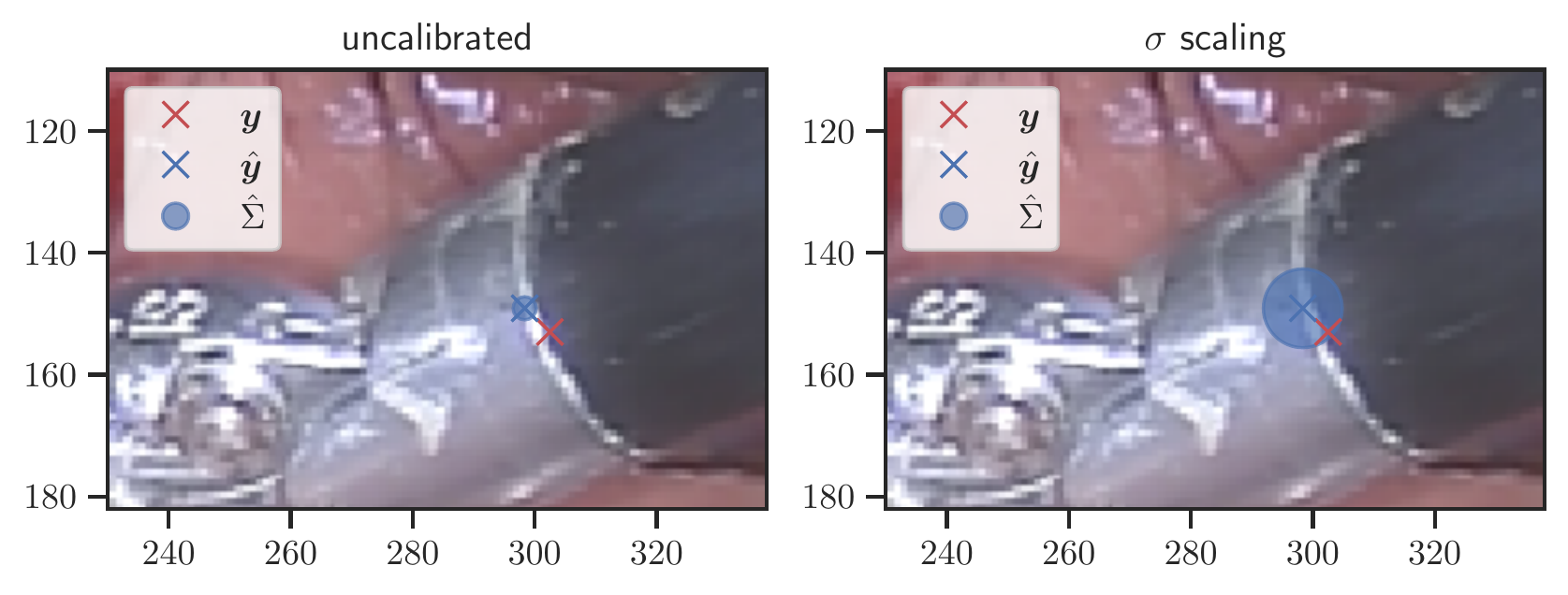} \\
    \small pixel coordinates
    \caption{Example result from EndoVis test set. The task is to predict pixel coordinates of the forceps shaft center. Before calibration, the uncertainty is underestimated and the true instrument position $ \bm{y} $ does not fall into the prediction region $ \hat{\bm{y}} \pm \hat{\Sigma} $. After calibration with $ \sigma $ scaling, the uncertainty better reflects the predictive error.}
    \label{fig:result_endovis}
\end{figure}

\begin{table}
    \caption{Uncertainty calibration error test set results for different datasets and model architectures (averaged over 5 runs). High UCE values indicate miscalibration. In addition, the resulting $ s $ for $ \sigma $ scaling is given. We also report UCE values for an ensemble of DenseNets. Bold font indicates lowest values in each experiment.}
    \scriptsize
    \centering
    \begin{tabular}{ccc|cccc|cccc||c}
    \toprule
     & \multicolumn{2}{c}{} & \multicolumn{4}{c}{Levi et al.} & \multicolumn{4}{c}{ours} & \\
    \cmidrule(lr){4-7} \cmidrule(lr){8-11}
    Data Set & Model & MSE & none & aux & $ \sigma $ & $ s $ & none & aux & $ \sigma $ & $ s $ & ensemble \\
    \midrule
                & ResNet-101      & 6.4e-3 & 0.51 & 0.35 & 0.28 & 2.91 & 0.49 & 0.31 & \textbf{0.20} & 2.37 & \\
    BreastPathQ & DenseNet-201    & 7.0e-3 & 0.21 & 0.38 & \textbf{0.15} & 1.62 & \textbf{0.11} & 0.36 & 0.15 & 1.33 & 0.51 \\
                & EfficientNet-B4 & 6.4e-3 & 0.49 & 0.65 & \textbf{0.10} & 2.30 & 0.46 & 0.47 & 0.17 & 1.77 \\
    \midrule
            & ResNet-101          & 5.3e-3 & 0.28 & 0.07 & 0.06 & 1.46 & 0.28 & \textbf{0.02} & 0.06 & 1.40 \\
    BoneAge & DenseNet-201        & 3.5e-3 & 0.31 & \textbf{0.05} & \textbf{0.05} & 2.98 & 0.31 & \textbf{0.05} & \textbf{0.05} & 2.54 & 0.09 \\
            & EfficientNet-B4     & 3.5e-3 & 0.30 & 0.05 & 0.10 & 4.83 & 0.30 & \textbf{0.03} & 0.12 & 3.98 \\
    \midrule
            & ResNet-101          & 4.0e-4 & \textbf{0.04} & 0.10 & 0.09 & 6.07 & \textbf{0.04} & \textbf{0.04} & \textbf{0.04} & 3.50 \\
    EndoVis & DenseNet-201        & 1.1e-3 & 0.09 & 0.05 & 0.05 & 3.24 & \textbf{0.04} & \textbf{0.04} & \textbf{0.04} & 2.57 & 0.08 \\
            & EfficientNet-B4     & 8.9e-4 & 0.06 & 0.05 & 0.06 & 2.25 & 0.06 & \textbf{0.04} & \textbf{0.04} & 1.79 \\
    \midrule
        & ResNet-101              & 2.0e-3 & 0.17 & 0.02 & 0.02 & 2.74 & 0.17 & \textbf{0.01} & 0.02 & 2.14 \\
    OCT & DenseNet-201            & 1.3e-3 & 0.08 & \textbf{0.01} & 0.02 & 1.60 & 0.04 & 0.03 & 0.02 & 1.26 & 0.67 \\
        & EfficientNet-B4         & 1.4e-3 & 0.12 & \textbf{0.01} & \textbf{0.01} & 2.65 & 0.12 & \textbf{0.01} & \textbf{0.01} & 1.94 \\
    \bottomrule
    \end{tabular}
    \label{tab:results}
\end{table}

\subsection{Posterior Prediction Intervals}
\label{sec:credible}

In addition to the calibration diagrams, we compute prediction intervals from the uncalibrated and calibrated posterior predictive distribution.
Well-calibrated prediction intervals provide a reliable measure of precision of the estimated target value.
In Bayesian inference, prediction intervals define an interval within which the true target value $ \bm{y}^{\ast} $ of a new, unobserved input $ \bm{x}^{\ast} $ is expected to fall with a specific probability \citep{Heskes1997,Held2014}.
This is also referred to as the credible interval of the posterior predictive distribution.
For $ \gamma \in (0, 1) $, a $ \gamma \cdot 100\,\% $ prediction interval is defined through $ z_{l} $ and $ z_{u} $ such that
\begin{equation}
    \int_{z_{l}}^{z_{u}} p(\bm{y}^{\ast} \,\vert\, \bm{x}^{\ast}, \mathcal{D}) \, \mathrm{d}\bm{y}^{\ast} = \gamma ~ ,
\end{equation}
with posterior predictive distribution $ p(\bm{y}^{\ast} \,\vert\, \bm{x}^{\ast}, \mathcal{D}) $.
We compute the 50\,\%, 90\,\%, 95\,\%, and 99\,\% prediction interval using the root of the predictive variance from Eq.\,(\ref{eq:pred_variance}); that is, the $ \hat{\bm{y}} \pm z \hat{\Sigma} $ intervals with $ z \in \{ \Phi (0.5), \Phi (0.9), \Phi (0.95), \Phi (0.99) \} $ (estimated interval), with probit function $ \Phi (p) = \sqrt{2} \mathrm{erf}^{-1}( p ) $ and $ \mathrm{erf}(p) $ is the Gaussian error function.
This assumes that the posterior predictive distribution is Gaussian, which is not generally the case.
To assess the calibration of the posterior prediction interval, we compute the percentage of how many of the ground truth values of the test set actually fall within the respective intervals (observed interval).
In Fig.\,\ref{fig:credible}, selected plots of observed vs.\ estimated prediction intervals are shown.
A complete list of prediction intervals can be found in Appendix~\ref{app:credible}.

In general, the uncalibrated prediction intervals are estimated to be too narrow, which is a direct consequence of the underestimated predictive variance.
For example, the uncalibrated 90\,\% interval on DenseNet-201/BoneAge actually only contains approx.\ 50\,\% of the ground truth values.
On this data set, the prediction intervals are considerably improved after recalibration (Fig.\,\ref{fig:credible} left).
If a network is already well-calibrated, recalibration can lead to overestimation of the lower prediction intervals (Fig.\,\ref{fig:credible} right).
However, in all cases, the 99\,\% prediction interval contains approx.\ 99\,\% of the ground truth test set values after recalibration.
This is not the case without the proposed calibration methods.
Fig.\,\ref{fig:result_endovis} shows a practical example of the $ \hat{\bm{y}} \pm \hat{\Sigma} $ prediction region from the EndoVis test set.
Even though the posterior predictive distribution is not necessarily Gaussian, the calibrated results fit the prediction intervals well.
This is especially the case for BoneAge, which is the largest data set used in this paper.

\begin{figure}[p]
    \centering 
    \includegraphics[scale=0.7]{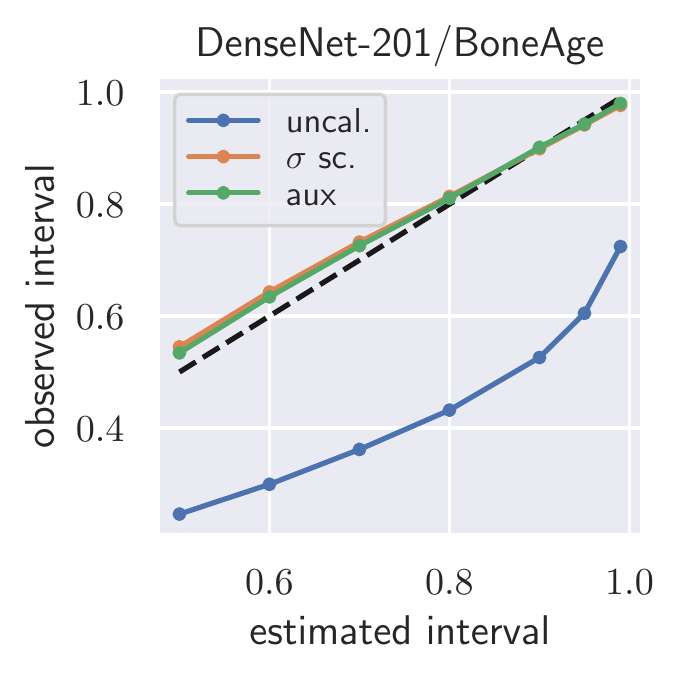}
    \includegraphics[scale=0.7]{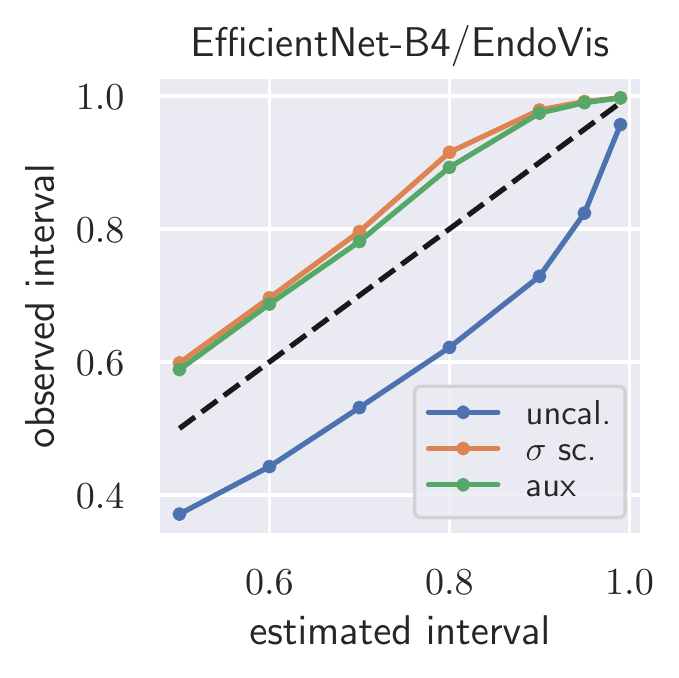}
    \includegraphics[scale=0.7]{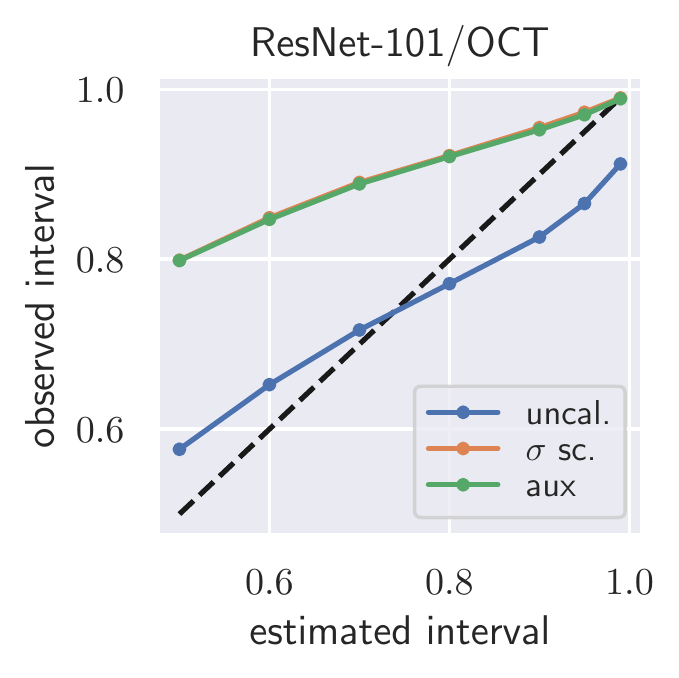}
    \caption{Observed vs.\ estimated posterior prediction intervals on the test sets. (Left \& center) The uncalibrated prediction interval is too narrow due to underestimation of uncertainty. (Right) Calibration can lead to overestimation of predictive intervals, if the network is already well-calibrated. Dashed lines denote 1:1 mapping.}
    \label{fig:credible}
\end{figure}

\subsection{Detection of Out-of-Distribution Data and Unreliable Predictions}

Deep neural networks only yield reliable predictions for data which follow the same distribution as the training data.
A shift in distribution could occur when a model trained on CT data from a specific CT device is applied to data from another manufacturer's CT device, for example.
This could potentially lead to wrong predictions with low uncertainty, which we tackle with recalibration.
To create a moderate distribution shift, we preprocess images from the BoneAge data set using Contrast Limited Adapative Histogram Equalization (CLAHE) \citep{Pizer1987} with a clip-limit of 0.03 and report histograms of the uncertainties (see Fig.\,\ref{fig:ood}).
Additionally, a severe distribution shift is created by presenting images from the BreastPathQ data set to the models trained on BoneAge.
\citet{Lakshminarayanan2017} state that deep ensembles provide better-calibrated uncertainty than Bayesian neural networks with MC dropout variational inference.
We therefore train an ensemble of 5 randomly initialized DenseNet-201 and compare Bayesian uncertainty with $ \sigma $ scaling to ensemble uncertainty under distribution shift.
 The results with $ \sigma $ scaling are comparable to those from a deep ensemble for a moderate shift, but without the need to train multiple models on the same data set.
A severe shift leads to only slightly increased uncertainties from the calibrated MC dropout model, while the deep ensemble is more sensitive.

Additionally, we apply the well-calibrated models to detect and reject uncertain predictions, as crucial decisions in medical practice should only be made on the basis of reliable predictions.
An uncertainty threshold $ \Sigma_{\mathrm{max}}^{2} $ is defined and all predictions from the test set are rejected where $ \hat{\Sigma}^{2} > \Sigma_{\mathrm{max}}^{2} $ (see Fig.\,\ref{fig:rejection}).
From this, a decrease in overall MSE is expected.
We additionally compare rejection on the basis of $ \sigma $ scaled uncertainty to uncertainty from the aforementioned ensemble.
In case of $ \sigma $ scaling, the test set MSE decreases monotonically as a function of the uncertainty threshold, whereas the ensemble initially shows an increasing MSE (see Fig.\,\ref{fig:rejection}).

\begin{figure}[p]
    \centering 
    \includegraphics[scale=.80]{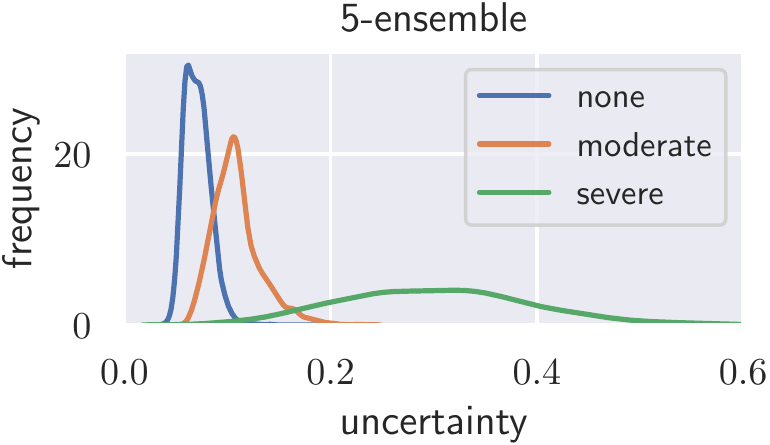} ~ \includegraphics[scale=.80]{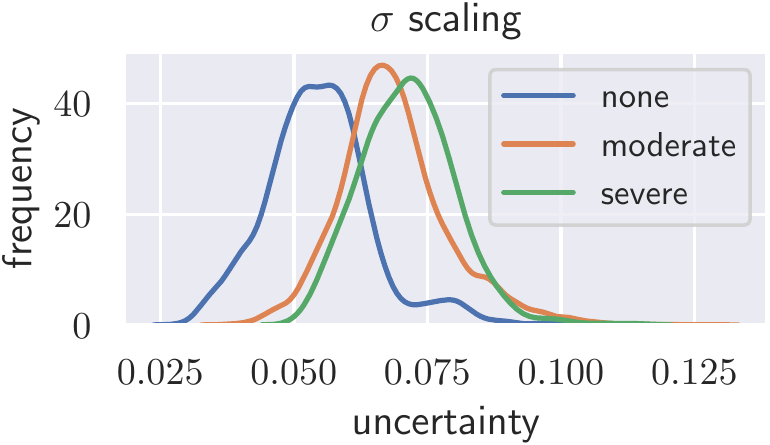}
    \vspace{-5mm}
    \caption{Histograms of the uncertainties for out-of-distribution detection with DenseNet-201 on BoneAge test set. (Left) Uncertainties from a non-Bayesian ensemble of five DenseNets and (right) Bayesian uncertainties calibrated with $ \sigma $ scaling. The distribution shifts have been created with pre-processing by CLAHE (moderate) and images from a different domain (severe).}
    \label{fig:ood}
\end{figure}

\begin{figure}[p]
    \centering
    \includegraphics[scale=.8]{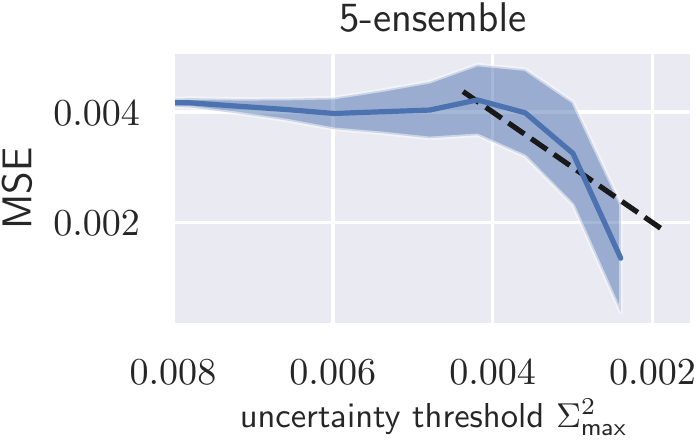} ~ \includegraphics[scale=.8]{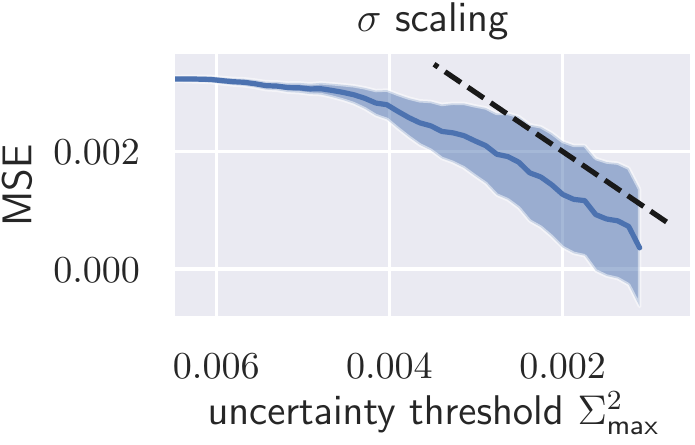}
    \caption{Rejection of uncertain predictions with DenseNet-201 on BoneAge test set with $ \hat{\Sigma}^{2} > \Sigma^{2}_{\mathrm{max}} $. The shaded area width visualizes the percentage of rejected samples. The dashed line visualizes linear relationship.}
    \label{fig:rejection}
\end{figure}

\section{Discussion \& Conclusion}


In this paper, well-calibrated predictive uncertainty in medical imaging obtained by variational inference with deep Bayesian models is discussed.
Both aux and $ \sigma $ scaling calibration methods considerably reduce miscalibration of predictive uncertainty in terms of UCE.
If the deep model is already well-calibrated, $ \sigma $ scaling does not negatively affect the calibration, which results in $ s \rightarrow 1 $.
More complex calibration methods such as aux scaling have to be used with caution, as they can overfit the data set used for calibration.
If the calibration set is sufficiently large, they can outperform simple scaling.
However, models trained on large data sets are generally better calibrated and the benefit diminishes.
Compared to the work of \citet{Levi2019}, accounting for epistemic uncertainty is particularly beneficial for smaller data sets, which is helpful in medical practice where access to large labeled data sets is less common and is associated with great costs.

Posterior prediction intervals provide another insight into the calibration of deep models.
After recalibration, the 99\,\% posterior prediction intervals correctly contain approx.\ 99\,\% of the ground truth test set values.
In some cases, lower prediction intervals are estimated to be too wide after calibration.
This is especially the case for smaller data sets and we conjecture that small calibration sets may not contain enough i.i.d.\ data for calibrating lower prediction intervals and that the assumption of a Gaussian predictive distribution is too strong in this case.
On the smallest data set BreastPathQ, aux scaling seems to perform better in terms of prediction intervals, but not in terms of UCE.

Well-calibrated uncertainties from MC dropout are able to detect a moderate shift in the data distribution.
However, deep ensembles perform better under a severe distribution shift.
BNNs with calibrated uncertainty by $ \sigma $ scaling outperform ensemble uncertainty in the rejection task, which we attribute to the generally poorer calibration of ensembles on in-distribution data.

$ \sigma $ scaling is simple to implement, does not change the predictive mean $ \hat{\bm{y}} $, and therefore guarantees to conserve the model's accuracy.
It is preferable to regularization (e.\,g.\ early stopping) or more complex recalibration methods in calibrated uncertainty estimation with Bayesian deep learning.
The disconnection between training and test NLL can successfully be closed, which creates highly accurate models with reliable uncertainty estimates.
However, there are many factors (e.\,g.\ network capacity, weight decay, dropout configuration) influencing the uncertainty that have not been discussed here and will be addressed in future work.

\acks{We thank Vincent Modes and Mark Wielitzka for their insightful comments. This research has received funding from the European Union as being part of the ERDF OPhonLas project.}

\sloppy
\bibliography{laves20.bib}

\appendix

\newpage

\section{Laplacian Model}

Using $ \mathsf{Laplace}(\hat{\bm{y}} (\bm{x}), \hat{\sigma} (\bm{x})) $ as model, the conditional log-likelihood is given by
\begin{align} 
      \log p(\bm{Y} \,\vert\, \bm{X}, \bm{\theta}) =& \sum_{i=1}^{m} \log \left( \frac{1}{2 \hat{\sigma}^{(i)}_{\bm{\theta}}} \exp \left\{ - \frac{\big| \bm{y}^{(i)} - \hat{\bm{y}}_{\bm{\theta}}^{(i)} \big|}{\hat{\sigma}^{(i)}_{\bm{\theta}}} \right\} \right) ~ ,
        \label{eq:laplacian_derive}
\end{align} 
which results in the following minimization criterion:
\begin{equation}
    \mathcal{L}_{\mathrm{L}}(\bm{\theta}) = \sum_{i=1}^{m} \frac{1}{\hat{\sigma}^{(i)}_{\bm{\theta}}} \big| \bm{y}^{(i)} - \hat{\bm{y}}_{\bm{\theta}}^{(i)} \big| + \log \big( \hat{\sigma}_{\bm{\theta}}^{(i)} \big) ~ .
    \label{eq:loss_laplacian}
\end{equation}
Using $ \mathcal{L}_{\mathrm{L}}(\bm{\theta}) $ instead of $ \mathcal{L}_{\mathrm{G}}(\bm{\theta}) $ results in applying an L1 metric on the predictive mean.
In some cases, this led to better results.
However, we have not conducted extensive experiments with it and leave it to future work.

\section{Derivation of \texorpdfstring{$\sigma$}{sigma} Scaling}

See §\,\ref{sec:scaling}.
Using a Gaussian model, we scale the standard deviation $ \sigma $ with a scalar value $ s $ to calibrate the probability density function
\begin{equation}
    p \left( \bm{y} \vert \bm{x} ; \hat{\bm{y}} (x), \hat{\sigma}^{2}(x)  \right) = \mathcal{N} \left( \bm{y} ; \hat{\bm{y}} (x), (s \cdot \hat{\sigma}(x))^{2} \right) ~ .
\end{equation}
The conditional log-likelihood is given by
\begin{align}
    \log p(\bm{Y} \,\vert\, \bm{X}, \bm{\theta}) &= \sum_{i=1}^{m} \log \left( \frac{1}{\sqrt{2\pi} s \hat{\sigma}_{\theta}^{(i)}} \exp \left( \frac{\big\Vert \bm{y}^{(i)} - \hat{\bm{y}}^{(i)}_{\bm{\theta}} \big\Vert^{2}}{2 \left( s \hat{\sigma}_{\bm{\theta}}^{(i)} \right)^{2}} \right) \right) \\
    &= - \dfrac{m}{2} \log\left( 2\pi \right) - \sum_{i=1}^{m} \log \left( s \hat{\sigma}^{(i)}_{\bm{\theta}} \right) + \frac{1}{2} \left( s \hat{\sigma}_{\bm{\theta}}^{(i)} \right)^{-2} \cdot \big\Vert \bm{y}^{(i)} - \hat{\bm{y}}_{\bm{\theta}}^{(i)} \big\Vert^{2}
\end{align}
This results in the following optimization objective (ignoring constants):
\begin{equation}
    \mathcal{L}_{\mathrm{G}}(s) = m \log(s) + \tfrac{1}{2} s^{-2} \sum_{i=1}^{m} (\hat{\sigma}_{\bm{\theta}}^{(i)})^{-2} \big\Vert \bm{y}^{(i)} - \hat{\bm{y}}_{\bm{\theta}}^{(i)} \big\Vert^{2} ~ .
    \label{eq:loss_scaling_gaussian_app}
\end{equation}
Using a Laplacian model, the optimization criterion follows as
\begin{equation}
    \mathcal{L}_{\mathrm{L}}(s) = m \log(s) + s^{-1} \sum_{i=1}^{m} \frac{1}{\hat{\sigma}_{\bm{\theta}}^{(i)}} \left| \bm{y}^{(i)} - \hat{\bm{y}}_{\bm{\theta}}^{(i)} \right| ~ .
    \label{eq:loss_scaling_laplacian_app}
\end{equation}
Eq.\,(\ref{eq:loss_scaling_gaussian_app}) and (\ref{eq:loss_scaling_laplacian_app}) are optimized w.r.t.\ $ s $ with fixed $ \bm{\theta} $ using gradient descent in a separate calibration phase after training.
The solution to Eq.\,(\ref{eq:loss_scaling_gaussian_app}) can also be written in closed form as
\begin{equation}
    s_{\mathrm{G}} = \pm \sqrt{\frac{1}{m} \sum_{i=1}^{m} \big( \hat{\sigma}_{\bm{\theta}}^{(i)} \big)^{-2} \big\Vert \bm{y}^{(i)} - \hat{\bm{y}}_{\bm{\theta}}^{(i)} \big\Vert^{2}}
\end{equation}
and the solution to Eq.\,(\ref{eq:loss_scaling_laplacian_app}) follows as
\begin{equation}
    s_{\mathrm{L}} = \frac{1}{m} \sum_{i=1}^{m} \frac{1}{\hat{\sigma}_{\bm{\theta}}^{(i)}} \left| \bm{y}^{(i)} - \hat{\bm{y}}_{\bm{\theta}}^{(i)} \right| ~ ,
\end{equation}
respectively.
We apply $ \sigma $ scaling to jointly calibrate aleatoric and epistemic uncertainty as described in §\,\ref{sec:well}.

\section{Unbiased Estimator of the Approximate Predictive Variance}
\label{app:proof_var}

We show that the expectation of the predictive sample variance from MC dropout, as given in \citep{Kendall2017}, equals the true variance of the approximate posterior predictive distribution.

\begin{prop}    
    Given $ N $ MC dropout samples $ \bm{f}_{\bm{\theta}_{n}} = [ \hat{\bm{y}}_{n}, \hat{\sigma}^{2}_{n} ] $ from our approximate predictive distribution $ p(\bm{y}^{\ast} \vert \bm{x}^{\ast}, \mathcal{D}) = \mathcal{N} ( \bm{y}^{\ast} ; \bm{y}, \Sigma^{2} ) $,
    the predictive sample variance
    \begin{equation}
    \hat{\Sigma}^{2} = \frac{1}{N} \sum_{n=1}^{N} \left( \hat{\bm{y}}_{n} - \frac{1}{N} \sum_{n=1}^{N} \hat{\bm{y}}_{n} \right)^{2} + \frac{1}{N} \sum_{n=1}^{N} \hat{\sigma}^{2}_{n}
\end{equation}
    is an unbiased estimator of the approximate predictive variance.
\end{prop}

\begin{proof}
\begin{align}
    \mathbb{E} \left[ \hat{\Sigma}^{2} \right] &= \mathbb{E} \left[ \frac{1}{N} \sum_{n=1}^{N} \left( \hat{\bm{y}}_{n} - \frac{1}{N} \sum_{n=1}^{N} \hat{\bm{y}}_{n} \right)^{2} + \frac{1}{N} \sum_{n=1}^{N} \hat{\sigma}^{2}_{n} \right] &\\
     &= \mathbb{E} \left[ \frac{1}{N} \sum_{n=1}^{N} \left( \hat{\bm{y}}_{n} - \frac{1}{N} \sum_{n=1}^{N} \hat{\bm{y}}_{n} \right)^{2} \right] + \mathbb{E} \left[ \frac{1}{N} \sum_{n=1}^{N} \hat{\sigma}^{2}_{n} \right] & \\
     &\mathrm{with} \quad \frac{1}{N} \sum_{n=1}^{N} \hat{\bm{y}}_{n} = \bar{\bm{y}} \quad \mathrm{follows} &\\
     &= \mathbb{E} \left[ \frac{1}{N} \sum_{n=1}^{N} \left( \hat{\bm{y}}_{n} - \bar{\bm{y}} \right)^{2} \right] + \hat{\sigma}^{2} &\\
     &= \mathbb{E} \left[ \frac{1}{N} \sum_{n=1}^{N} \left( \hat{\bm{y}}_{n} - \bar{\bm{y}} \right)^{2} + \bar{\bm{y}}^{2} - \bar{\bm{y}}^{2} + \bm{y}^{2} - \bm{y}^{2} + 2 \bar{\bm{y}}\bm{y} - 2 \bar{\bm{y}}\bm{y} \right] + \hat{\sigma}^{2} &\\
     &= \mathbb{E} \left[ \frac{1}{N} \sum_{n=1}^{N} \left( \hat{\bm{y}}_{n} - \bm{y} \right)^{2} - \left( \bar{\bm{y}} - \bm{y} \right)^{2} \right] + \hat{\sigma}^{2} &\\
     &= \mathbb{E} \left[
     \left( \hat{\bm{y}} - \bm{y} \right)^{2} \right] - \mathbb{E} \left[ \left( \bar{\bm{y}} - \bm{y} \right)^{2} \right] + \hat{\sigma}^{2} &\label{eq:proof_bias}\
     &= \Sigma^{2} - \hat{\sigma}^{2} + \hat{\sigma}^{2} &\\
     \mathbb{E}\left[ \hat{\Sigma}^{2} \right] &= \Sigma^{2} &
\end{align}
Note that the predicted heteroscedastic aleatoric uncertainty $ \hat{\sigma}^{2} $ equals the bias $ \mathbb{E} [ ( \bar{\bm{y}} - \bm{y} )^{2} ] $ in Eq.\,(\ref{eq:proof_bias}) when the aleatoric uncertainty is perfectly calibrated, thus $ \mathbb{E} [ ( \bar{\bm{y}} - \bm{y} )^{2} ] = \hat{\sigma}^{2} $.
\end{proof}

\section{Training Procedure}
\label{app:training}

The model implementations from PyTorch 1.3 \citep{PyTorch2019} are used and trained with the following settings:
\begin{itemize}
    \itemsep0em
    \item training for 500 epochs with batch size of 16
    \item Adam optimizer with initial learn rate of $ 3 \cdot 10^{-4} $ and weight decay with $ \lambda = 10^{-7} $
    \item reduce-on-plateau learn rate scheduler (patience of 20 epochs) with factor of 0.1
    \item in MC dropout, $ N=25 $ forward passes were performed with dropout with $ p = 0.5 $ used for ResNet (as described in \citep{Gal2016}). In DenseNet ($ p = 0.2 $) and EfficientNet ($ p = 0.4 $) standard dropout $ p $ of the architecture is used.
    \item Additional validation and test sets are used if provided by the data sets; otherwise, a train/validation/test split of approx. 50\%\,/\,25\%\,/\,25\% is used
    \item Source code for all experiments is available at \href{https://github.com/mlaves/well-calibrated-regression-uncertainty}{github.com/mlaves/well-calibrated-regression-uncertainty}
\end{itemize}

\section{3D OCT Needle Pose Data Set}
\label{app:dataset}

\begin{figure}[h]
    \centering
    \includegraphics[scale=1.0]{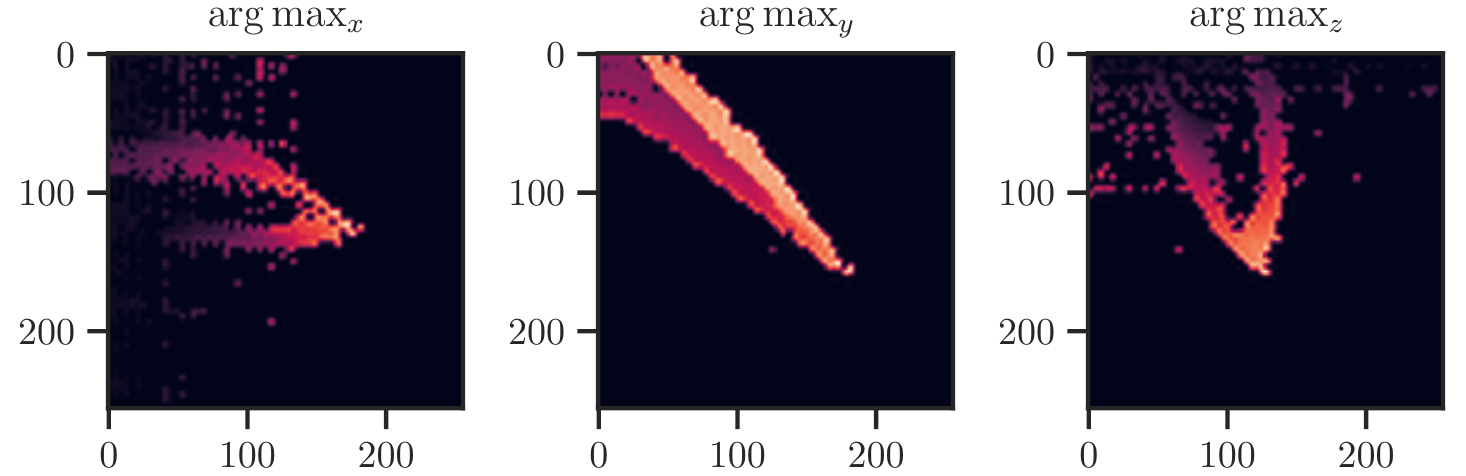} \\
    \small pixel coordinates
    \caption{Example image from OCT data set showing $ \argmax $ projections of a surgical needle tip acquired by optical coherence tomography.}
    \label{fig:oct_dataset}
\end{figure}

\noindent Our data set was created by attaching a surgical needle to a high-precision six-axis hexapod robot (H-826, Physik Instrumente GmbH \& Co. KG, Germany) and observing the needle tip with 3D optical coherence tomography (OCS1300SS, Thorlabs Inc., USA).
The data set consists of 5,000 OCT acquisitions with $ (64 \times 64 \times 512) $ voxels, covering a volume of approx.\ $ (3 \times 3 \times 3) $\,$\mathrm{mm}^{3}$.
Each acquisition is taken at a different robot configuration and labeled with the corresponding 6DoF pose $ \bm{y} \in \mathbb{R}^{6} $.
To process the volumetric data with CNNs for planar images, we calculate 3 planar projections along the spatial dimensions using the $ \argmax $ operator, scale them to equal size and stack them together as three-channel image (see Fig.\,\ref{fig:oct_dataset}).
A similar approach was presented in \citep{Laves2017} and \citep{Gessert2018}.
The data are characterized by a high amount of speckle noise, which is a typical phenomenon in optical coherence tomography.
The data set is publicly available at \href{https://github.com/mlaves/3doct-pose-dataset}{github.com/mlaves/3doct-pose-dataset}.

\section{Ablation Study on Auxiliary Model Scaling}
\label{app:aux_ablation}

\begin{figure}[h]
    \centering
    \includegraphics[scale=0.9]{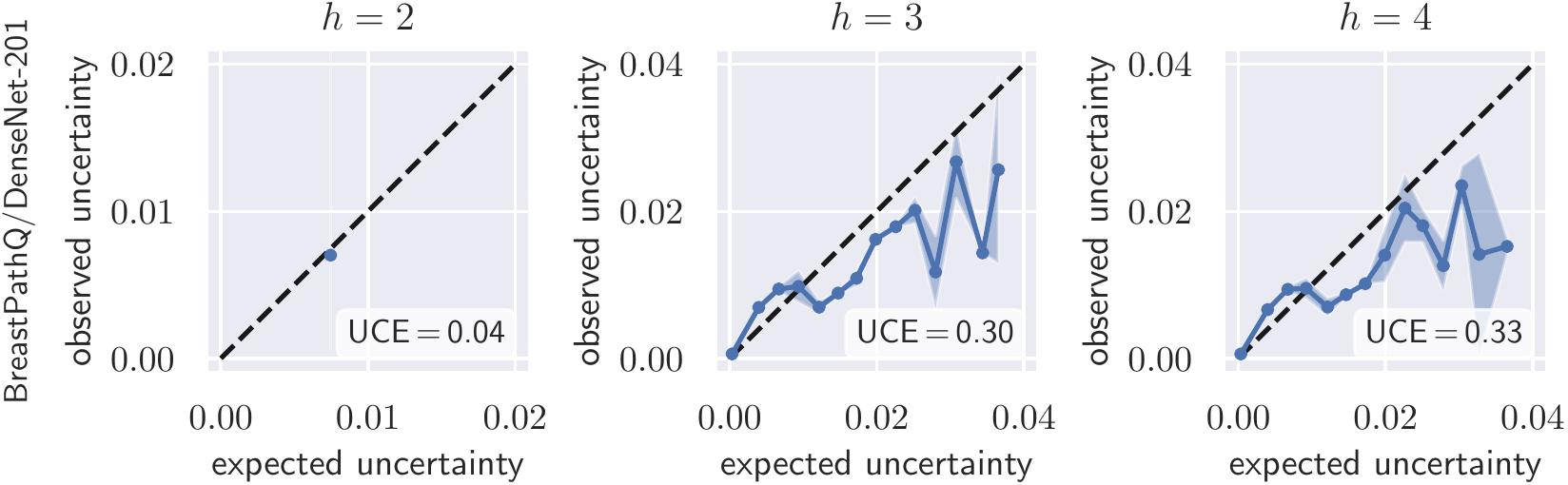}
    \caption{Calibration diagrams for aux scaling with different number of hidden layer units $ h $ on BreastPathQ/DenseNet-201.}
    \label{fig:aux_ablation}
\end{figure}

Here, we investigate the overfitting behavior of aux scaling by reducing the number of hidden layer units $ h $ of the two-layer auxiliary model with ReLU activations.
Aux scaling is more powerful than $ \sigma $ scaling, which can lead to overfitting the calibration set.
Fig.\,\ref{fig:aux_ablation} shows calibration diagrams for the auxiliary model ablations.
Reducing $ h $ leads to a minor calibration improvement, but at $ h=2 $, the model outputs a constant uncertainty, which is close to the overall mean of the observed uncertainty.
A single-layer single-unit model without bias would be equivalent to $ \sigma $ scaling.

\section{Additional Results and Calibration Diagrams}
\label{app:figures}

All test set runs have been repeated 5 times.
Solid lines denote mean and shaded areas denote standard deviation calculated from the repeated runs.

\begin{table}[h]
    \caption{Negative log-likelihood test set results for different datasets and model architectures (averaged over 5 runs). High NLL values indicate miscalibration. We also report NLL values for an ensemble of DenseNets. Bold font indicates lowest values in each experiment.}
    \scriptsize
    \centering
    \begin{tabular}{ccc|ccc|ccc||c}
    \toprule
     & \multicolumn{2}{c}{} & \multicolumn{3}{c}{Levi et al.} & \multicolumn{3}{c}{ours} & \\
    \cmidrule(lr){4-6} \cmidrule(lr){7-9}
    Data Set & Model & MSE & none & aux & $ \sigma $ & none & aux & $ \sigma $ & ensemble \\
    \midrule
                & ResNet-101      & 6.4e-3 & -0.78 & -5.06 & -5.06 & -2.89 & \textbf{-5.17} & -5.16 &  \\
    BreastPathQ & DenseNet-201    & 7.0e-3 & -5.16 & -5.84 & -5.70 & -5.67 & \textbf{-6.03} & -5.78 & 0.11 \\
                & EfficientNet-B4 & 6.4e-3 & -3.11 & -5.99 & -5.53 & -4.73 & \textbf{-6.16} & -5.62 &  \\
    \midrule
            & ResNet-101          & 5.3e-3 & -3.90 & \textbf{-4.34} & \textbf{-4.34} & -3.99 & \textbf{-4.34} & \textbf{-4.34} &  \\
    BoneAge & DenseNet-201        & 3.5e-3 &  1.74 & \textbf{-4.70} & -4.69 & -0.75 & \textbf{-4.70} & -4.69 & 0.07 \\
            & EfficientNet-B4     & 3.5e-3 & 13.61 & -4.74 & -4.67 &  6.40 & \textbf{-4.75} & -4.64 &  \\
    \midrule
            & ResNet-101          & 4.0e-4 & -0.53 & -6.32 & -6.33 & -3.85 & \textbf{-6.76} & -6.72 &  \\
    EndoVis & DenseNet-201        & 1.1e-3 & -0.72 & \textbf{-6.10} & -5.99 & -4.94 & -6.05 & -6.04 & 0.04 \\
            & EfficientNet-B4     & 8.9e-4 & -5.10 & -6.06 & -6.07 & -5.94 & \textbf{-6.17} & \textbf{-6.17} &  \\
    \midrule
        & ResNet-101              & 2.0e-3 & -1.08 & \textbf{-5.24} & \textbf{-5.24} & -3.38 & \textbf{-5.24} & \textbf{-5.24} &  \\
    OCT & DenseNet-201            & 1.3e-3 & -5.05 & -5.61 & -5.61 & -5.51 & \textbf{-5.62} & -5.61 & 0.10 \\
        & EfficientNet-B4         & 1.4e-3 & -1.72 & \textbf{-5.58} & -5.57 & -4.25 & \textbf{-5.58} & -5.57 &  \\
    \bottomrule
    \end{tabular}
    \label{tab:results_nll}
\end{table}

\begin{figure}[h]
    \centering
    \includegraphics[scale=0.9]{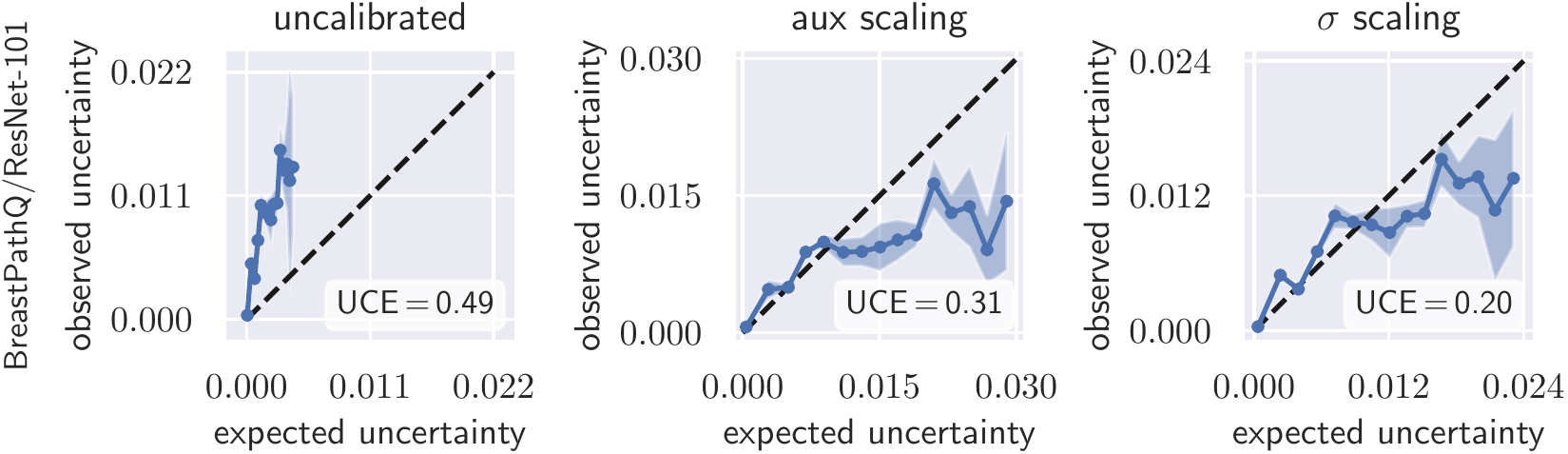}
\end{figure}

\begin{figure}[h]
    \centering
    \includegraphics[scale=0.9]{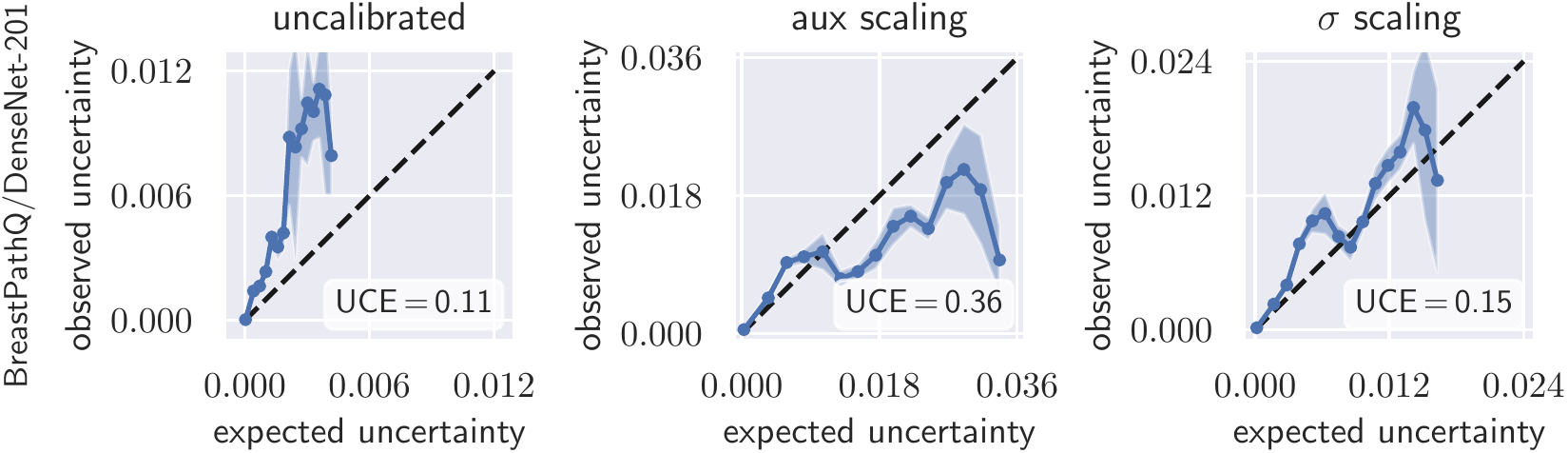}
\end{figure}

\begin{figure}[h]
    \centering
    \includegraphics[scale=0.9]{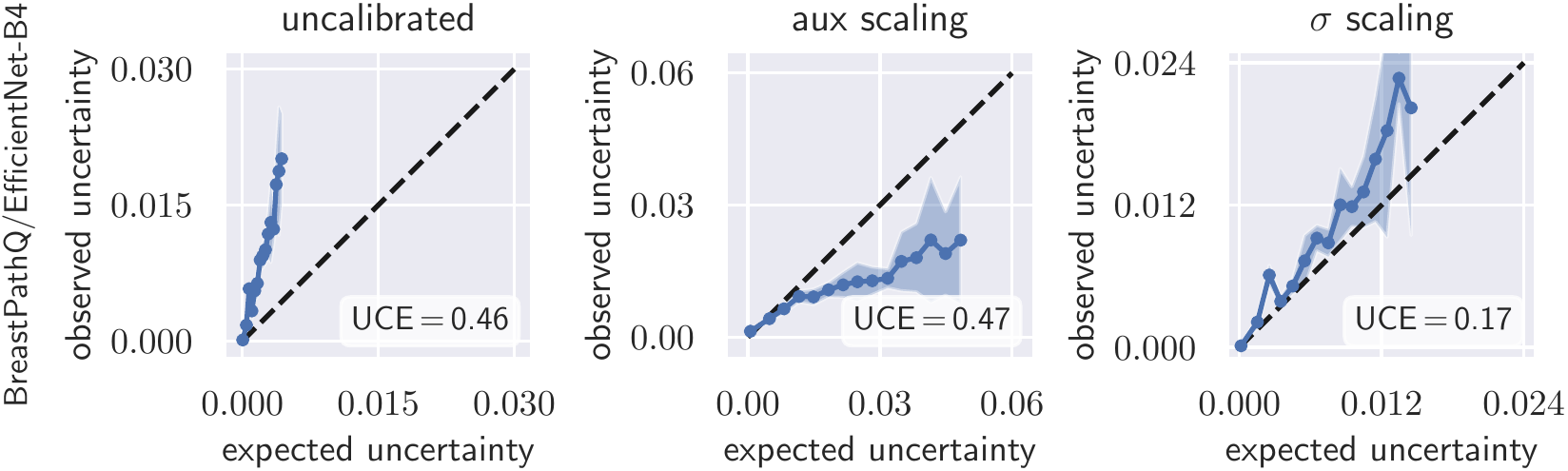}
\end{figure}

\pagebreak

\begin{figure}[h]
    \centering
    \includegraphics[scale=0.9]{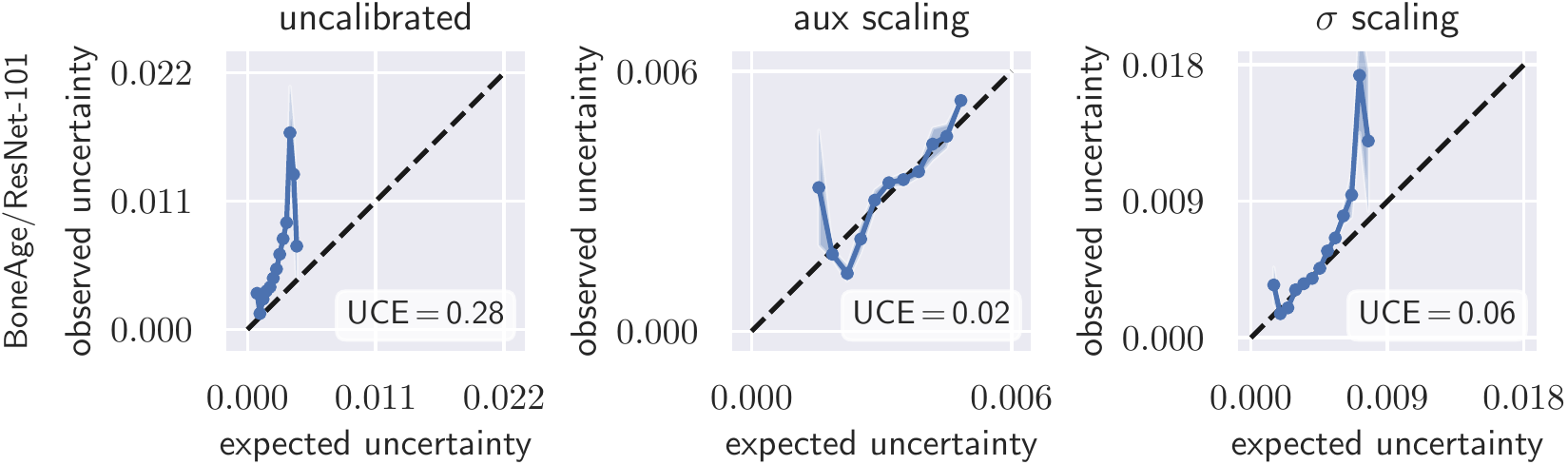}
\end{figure}

\begin{figure}[h]
    \centering
    \includegraphics[scale=0.9]{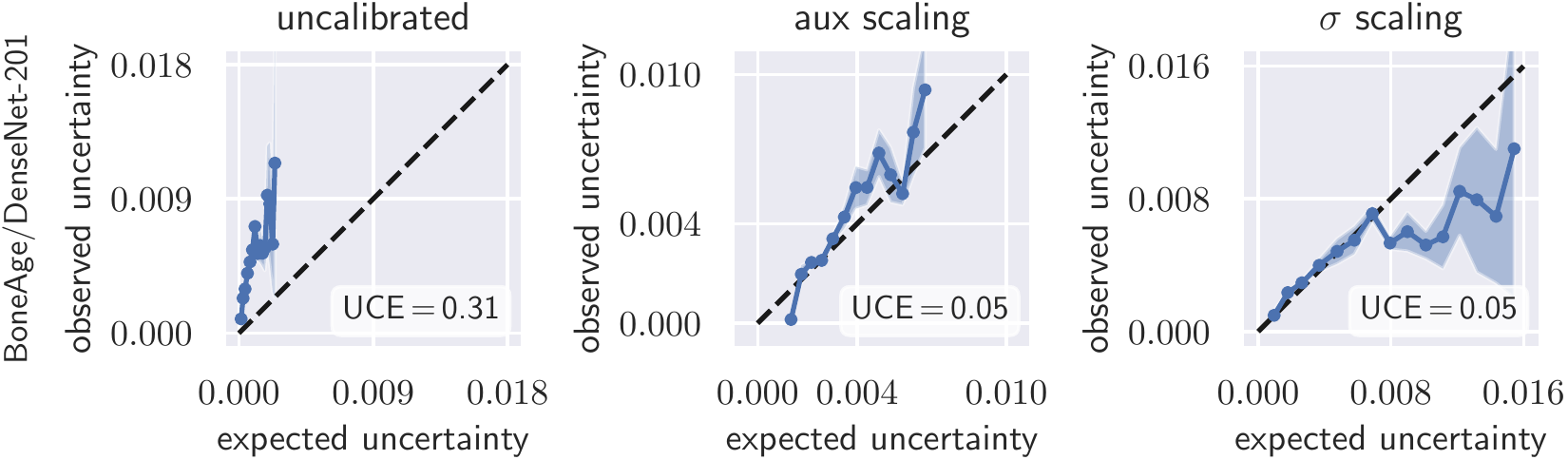}
\end{figure}

\begin{figure}[h]
    \centering
    \includegraphics[scale=0.9]{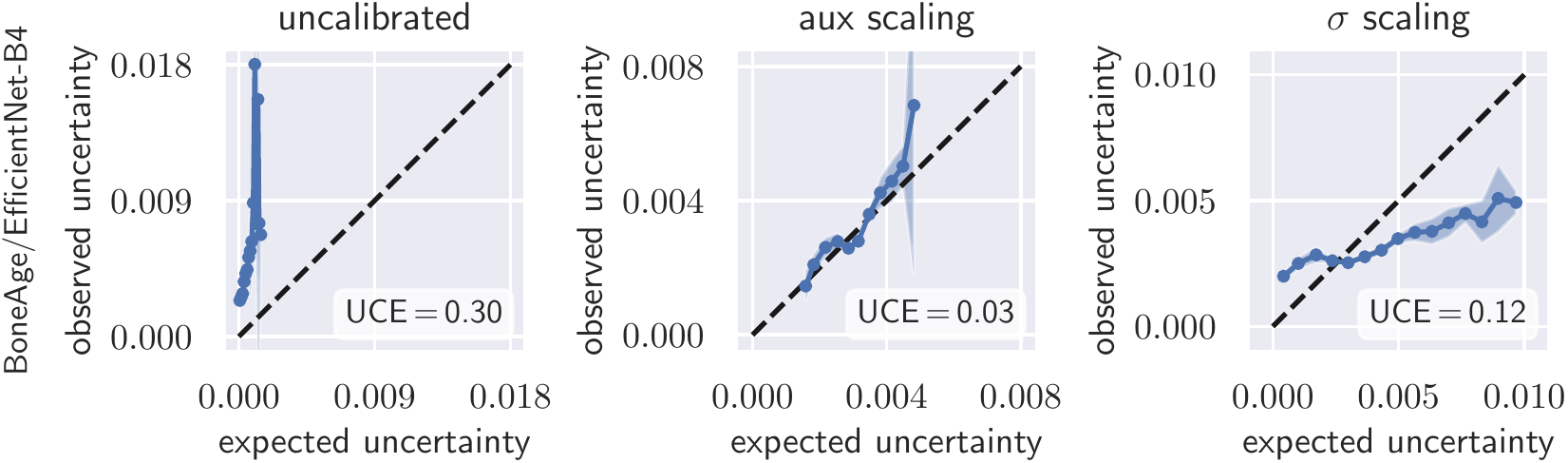}
\end{figure}

\pagebreak

\begin{figure}[h]
    \centering
    \includegraphics[scale=0.9]{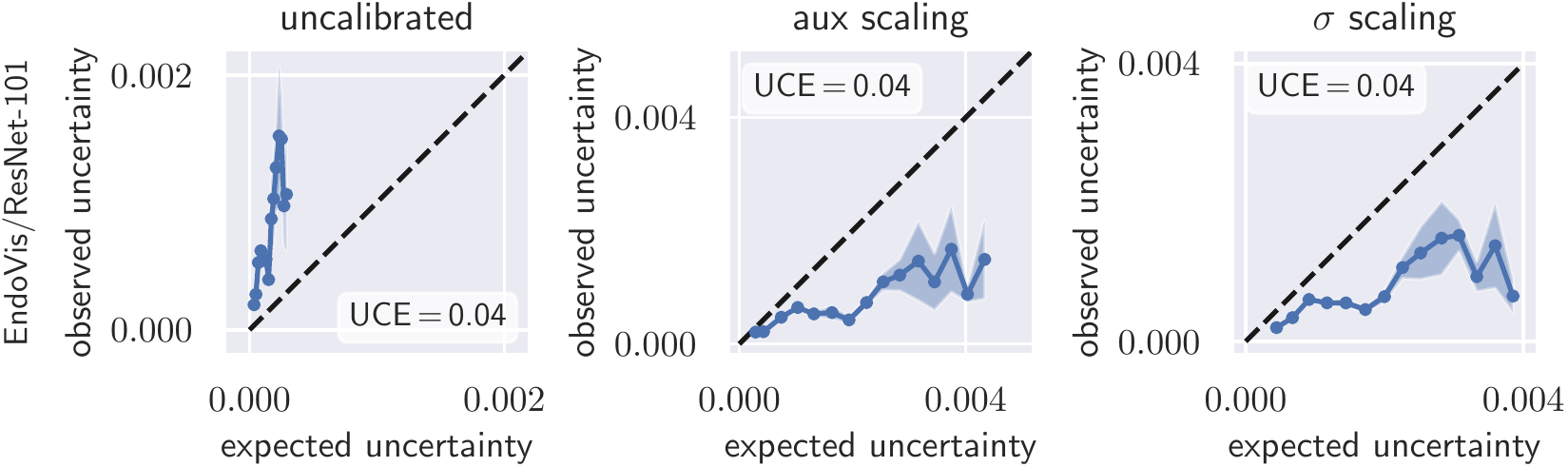}
\end{figure}

\begin{figure}[h]
    \centering
    \includegraphics[scale=0.9]{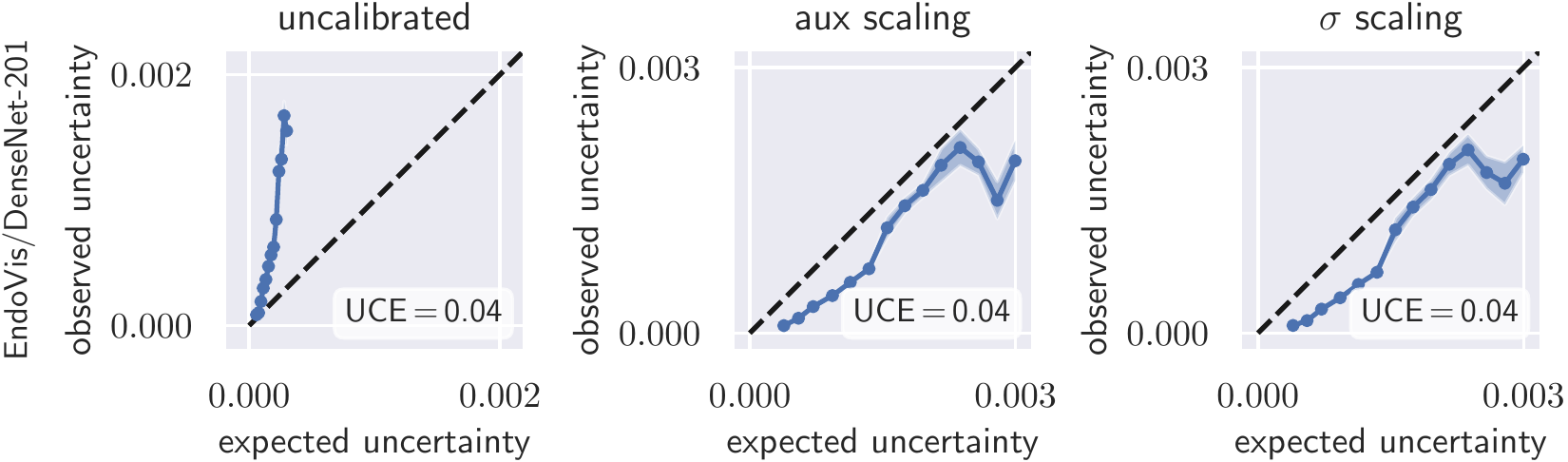}
\end{figure}

\begin{figure}[h]
    \centering
    \includegraphics[scale=0.9]{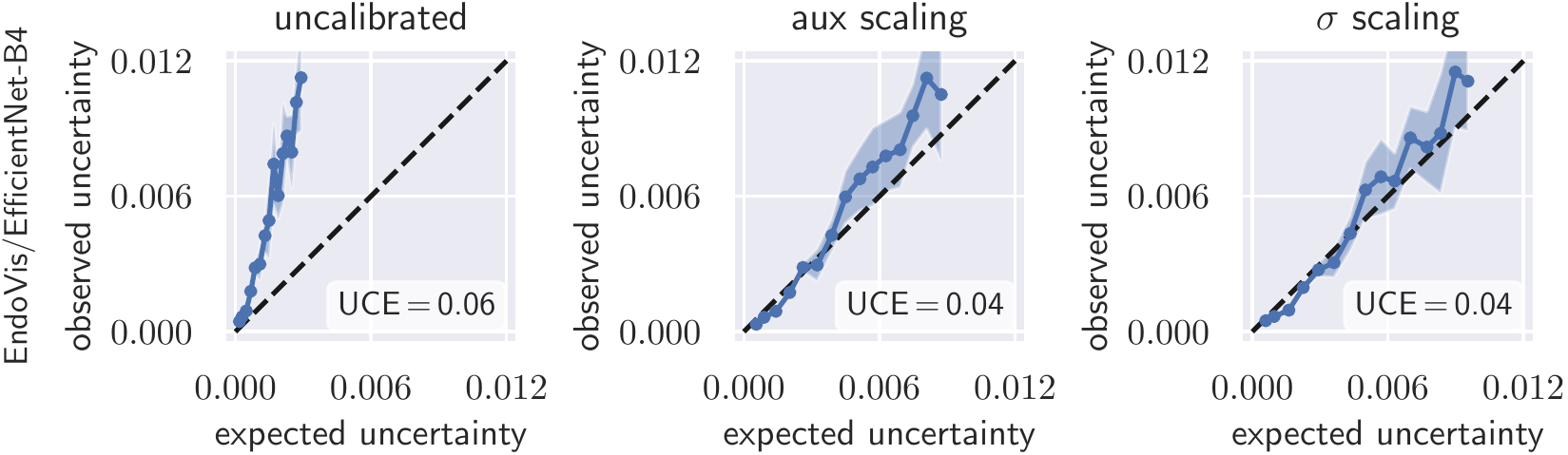}
\end{figure}

\pagebreak

\begin{figure}[h]
    \centering
    \includegraphics[scale=0.9]{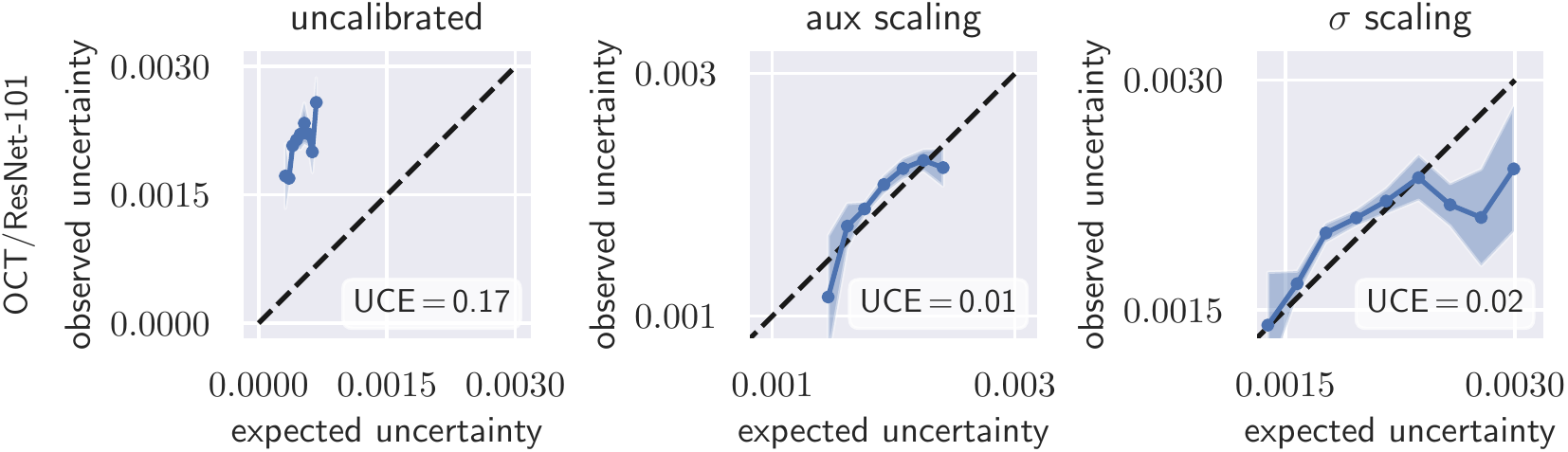}
\end{figure}

\begin{figure}[h]
    \centering
    \includegraphics[scale=0.9]{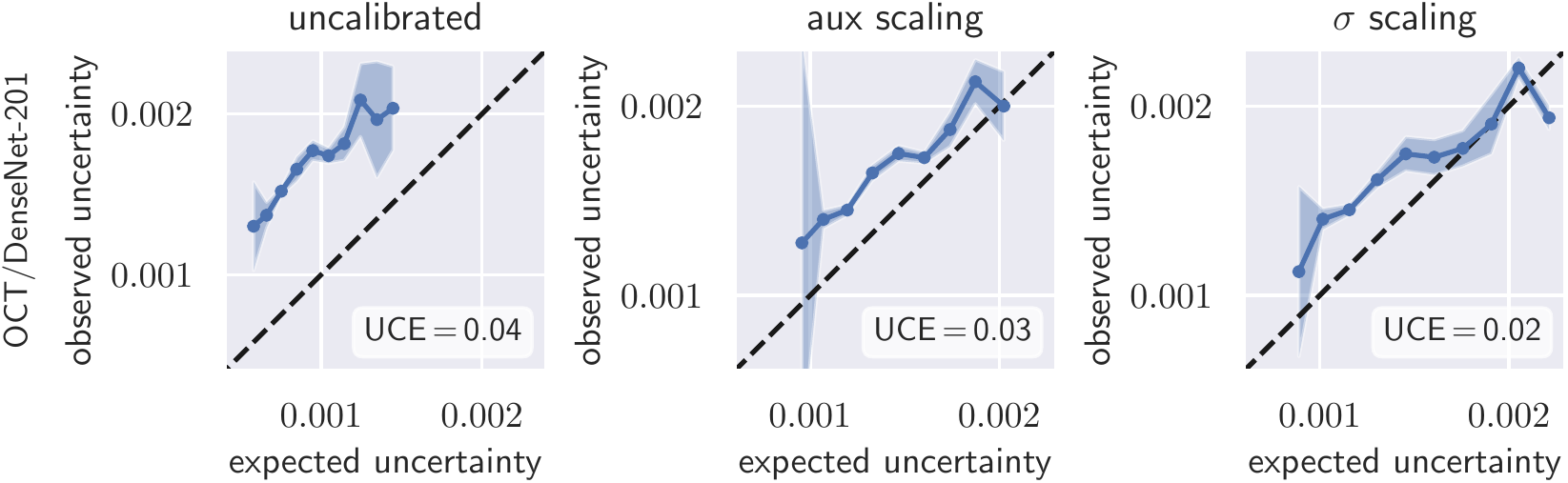}
\end{figure}

\begin{figure}[h]
    \centering
    \includegraphics[scale=0.9]{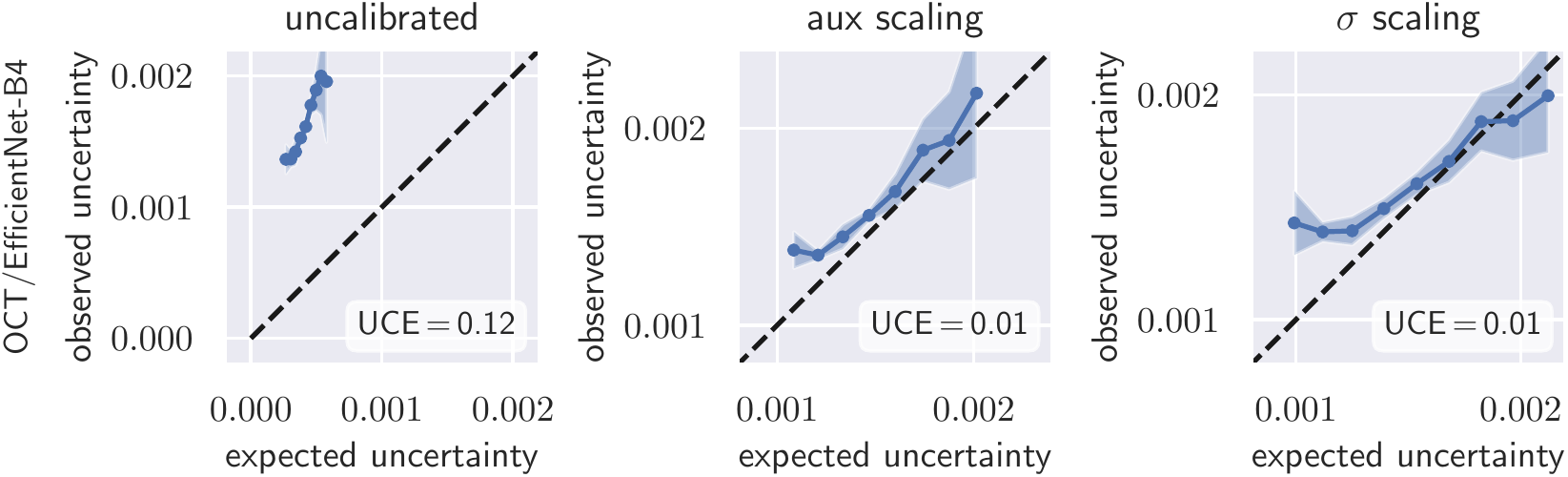}
\end{figure}
\clearpage

\subsection{Additional Prediction Intervals}
\label{app:credible}

\begin{figure}[h]
    \centering
    \includegraphics[scale=0.65]{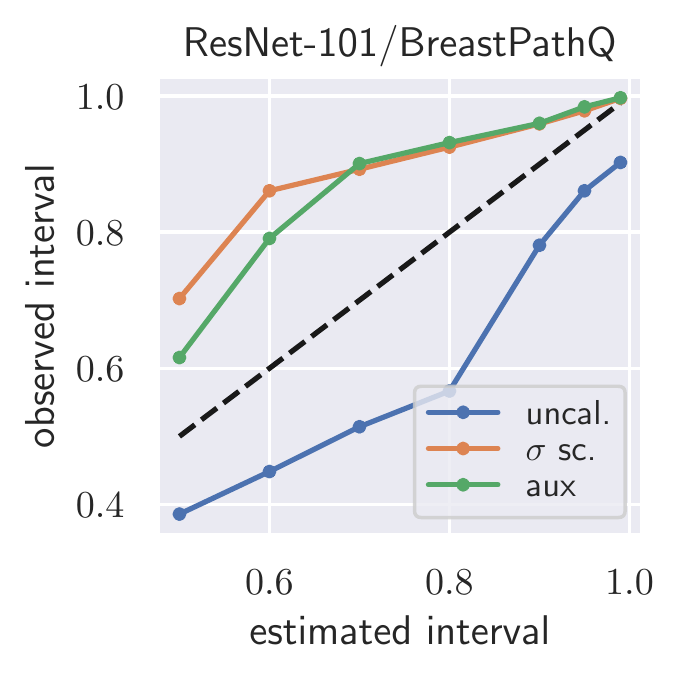}
    \includegraphics[scale=0.65]{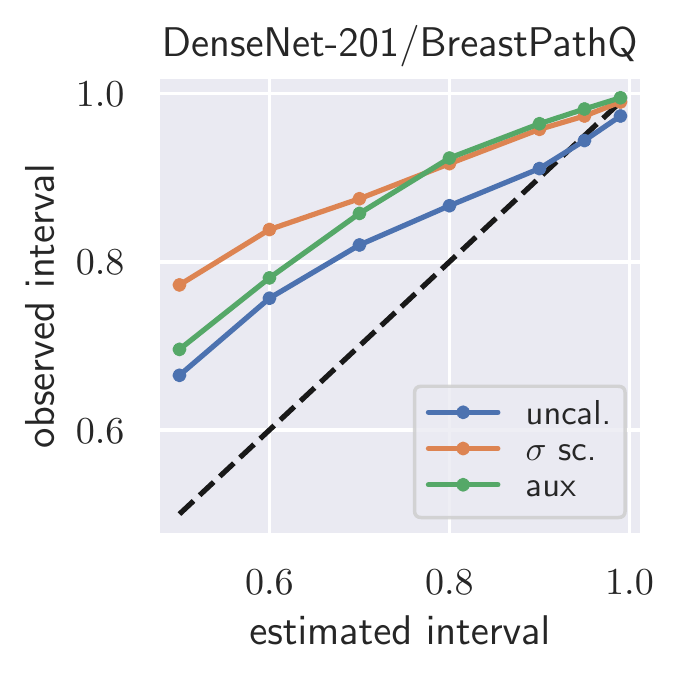}
    \includegraphics[scale=0.65]{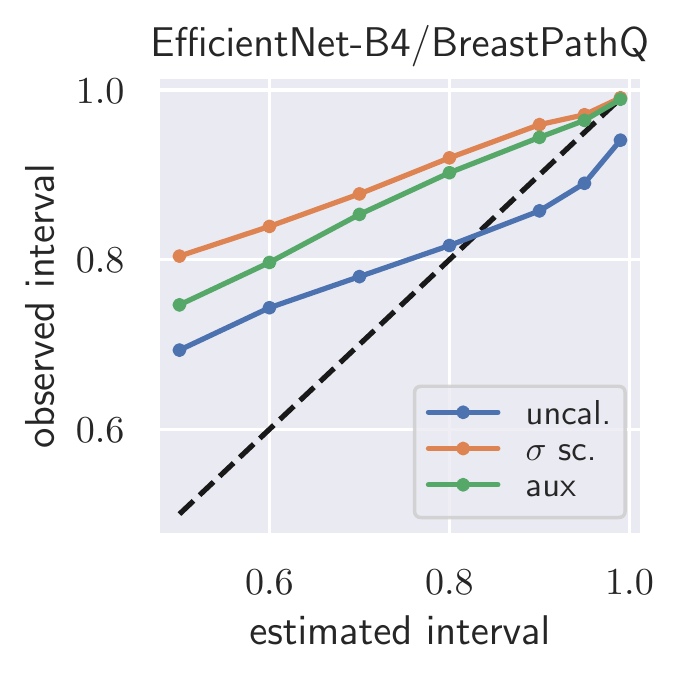} \\
    \includegraphics[scale=0.65]{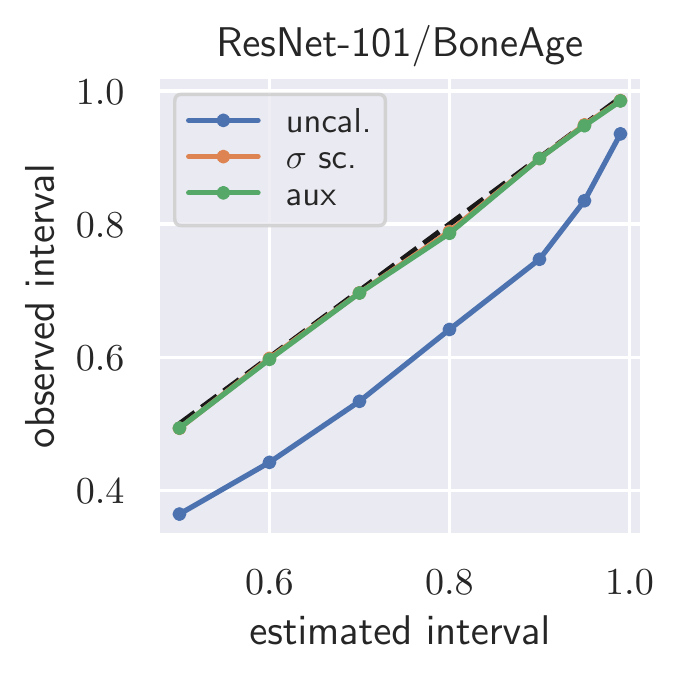}
    \includegraphics[scale=0.65]{prediction_interval_BoneAge_densenet201.pdf}
    \includegraphics[scale=0.65]{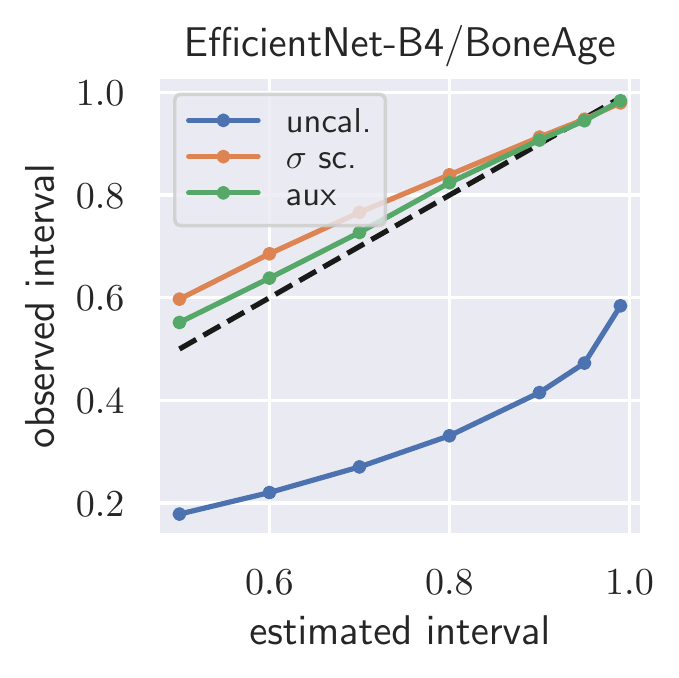} \\
    \includegraphics[scale=0.65]{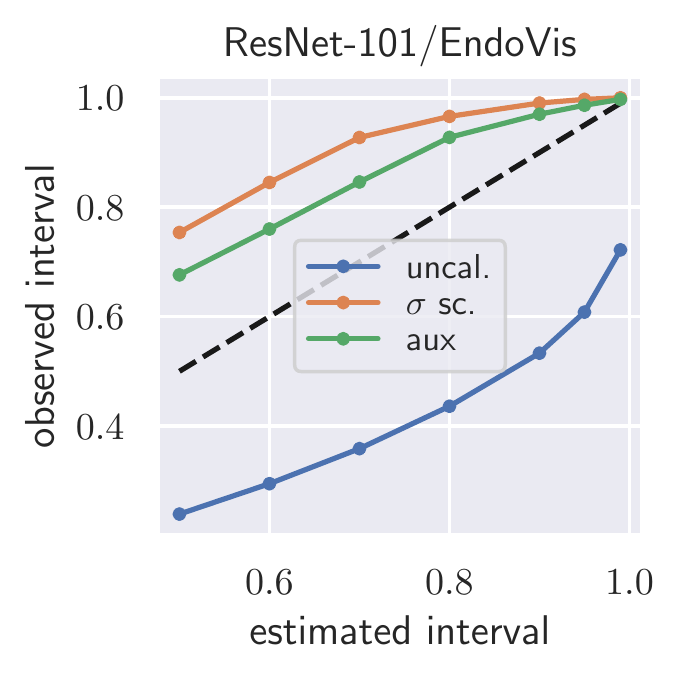}
    \includegraphics[scale=0.65]{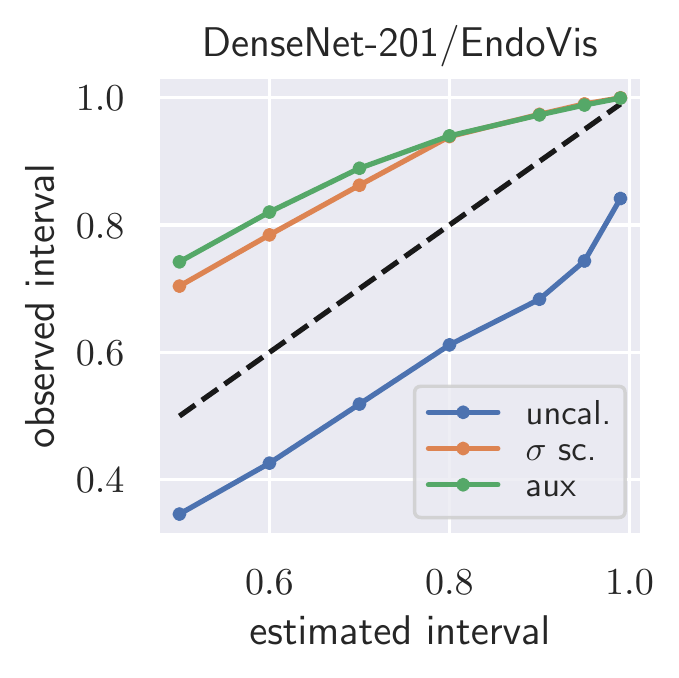}
    \includegraphics[scale=0.65]{prediction_interval_EndoVis_efficientnetb4.pdf} \\
    \includegraphics[scale=0.65]{prediction_interval_OCT_resnet101.pdf}
    \includegraphics[scale=0.65]{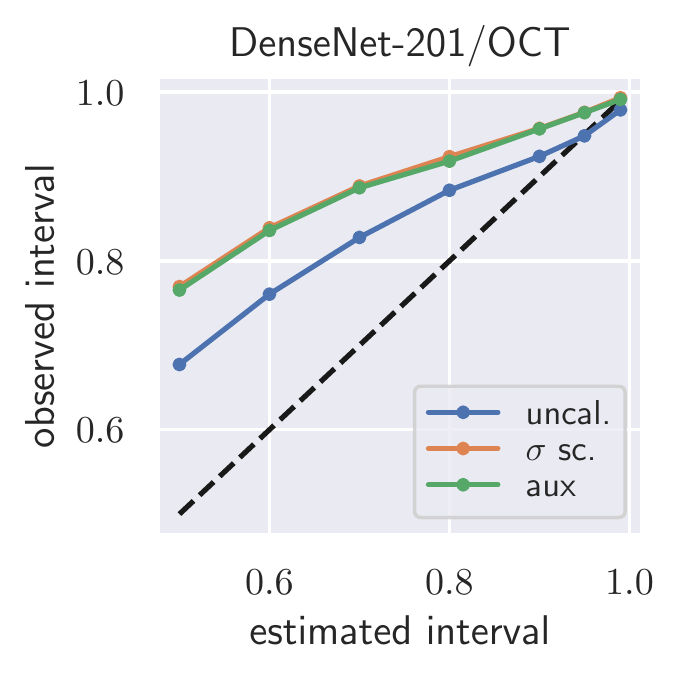}
    \includegraphics[scale=0.65]{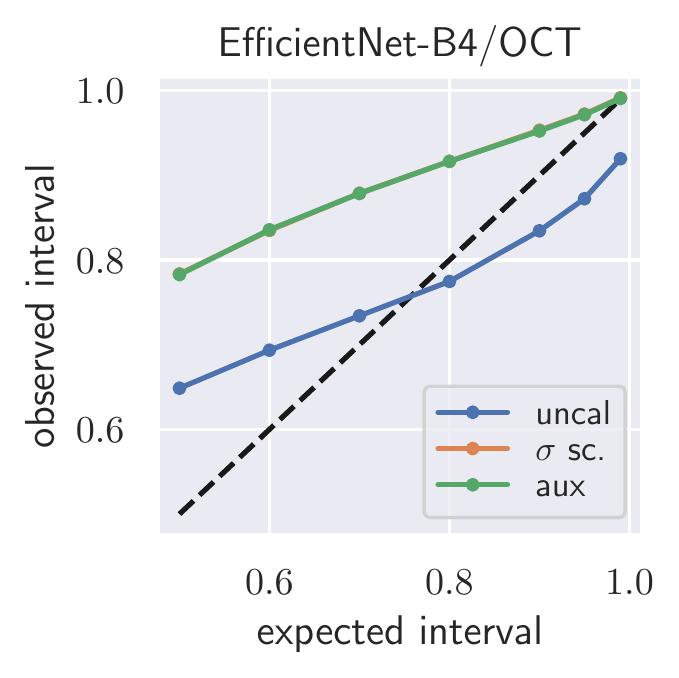}
    \caption{Observed vs.\ estimated posterior prediction intervals for all networks/data sets.}
    \label{fig:credible_app}
\end{figure}

\end{document}